\documentclass[12pt]{article}

\usepackage{amsmath}

\usepackage[english]{babel}
\usepackage[ansinew]{inputenc}
\usepackage{epsfig}
\usepackage{color,graphicx}
\usepackage{wrapfig}
\usepackage{amssymb,amsthm,amsmath}
\usepackage{dsfont}
\usepackage{natbib}
\usepackage{abstract}
\usepackage{enumerate}
\usepackage{comment}
\usepackage{algorithm}
\usepackage[noend]{algpseudocode}

\usepackage{caption}

\newcommand{\blind}{0}

\renewcommand{\baselinestretch}{1}

\newtheorem{proposition}{Proposition}[]

\algrenewcommand\algorithmicrequire{\textbf{Input:}}
\algrenewcommand\algorithmicensure{\textbf{Output:}}

\def\ds{\displaystyle}

\begin{document}








\def\spacingset#1{\renewcommand{\baselinestretch}%
{#1}\small\normalsize} \spacingset{1}


\if0\blind
{
  \title{\bf Exact Bayesian inference for level-set Cox processes with piecewise constant intensity function}
  \author{Fl\'{a}vio B. Gon\c{c}alves and B\'{a}rbara C. C. Dias \\
    \\
    Universidade Federal de Minas Gerais, Brazil}
    \date{}
  \maketitle
} \fi

\if1\blind
{
  \bigskip
  \bigskip
  \bigskip
  \begin{center}
    {\LARGE\bf Exact Bayesian inference for level-set Cox processes with piecewise constant intensity function}
\end{center}
  \medskip
} \fi

\bigskip

\begin{abstract}

This paper proposes a new methodology to perform Bayesian inference for a class of multidimensional Cox processes in which the intensity function is piecewise constant. Poisson processes with piecewise constant intensity functions are believed to be suitable to model a variety of point process phenomena and, given its simpler structure, are expected to provide more precise inference when compared to processes with non-parametric and continuously varying intensity functions. The partition of the space domain is flexibly determined by a level-set function of a latent Gaussian process. Despite the intractability of the likelihood function and the infinite dimensionality of the parameter space, inference is performed exactly, in the sense that no space discretization approximation is used and MCMC error is the only source of inaccuracy. That is achieved by using retrospective sampling techniques and devising a pseudo-marginal infinite-dimensional MCMC algorithm that converges to the exact target posterior distribution. Computational efficiency is favored by considering a nearest neighbor Gaussian process, allowing for the analysis of large datasets. An extension to consider spatiotemporal models is also proposed. The efficiency of the proposed methodology is investigated in simulated examples and its applicability is illustrated in the analysis of some real point process datasets.


\end{abstract}

\noindent%
\vfill

\section{Introduction}\label{sec_intro}

Point pattern statistical models aim at modeling the occurrence of a given event of interest in some region. This is often a compact region is $\mathds{R}^2$ such that each data point is interpreted as the location of occurrence of a given event of interest. The most widely used point process model is the Poisson process (PP), in which the number of events in any region has Poisson distribution and is independent for disjoint regions. The Poisson process dynamics is mainly determined by its intensity function (IF) which, roughly speaking, determines the instant rate of occurrence of the event of interest across the region being considered. If the IF is assumed to vary stochastically, the resulting process is called a Cox process. Several classes of Cox process models have already been proposed in the literature, including non-parametric models in which the IF varies continuously as a function of a latent Gaussian process \citep{moller1998log,G&G}. For several of the real examples considered to fit those models, inference results suggest that a piecewise constant IF ought to be suitable to accommodate the variability of the observed process. Figure \ref{fig1} shows three examples of estimated intensity functions regarding white oaks in Lansing Woods, USA, particles in a bronze filter section profile, and fires in a region of New Brunswick, Canada. All the datasets are available in the R package \texttt{spatstat} \citep{baddeley2015spatial} and are revisited in the analyzes presented in Section \ref{sec_app}. The IF estimates are obtained via kernel smoothing using the R package \texttt{splancs} \citep{rowlingson2012package} through the function \emph{kernel2d}. Results suggest that a piecewise constant IF assuming up to five different values should be suitable to fit those datasets. This is based on the variance behavior of the Poisson process given its IF which, in turn, is based on the variance of the Poisson distribution. Other motivating examples can be found in \citet{hildeman2017level}.

\begin{figure}[h!]
	\centering{
		\includegraphics[width=0.9\linewidth]{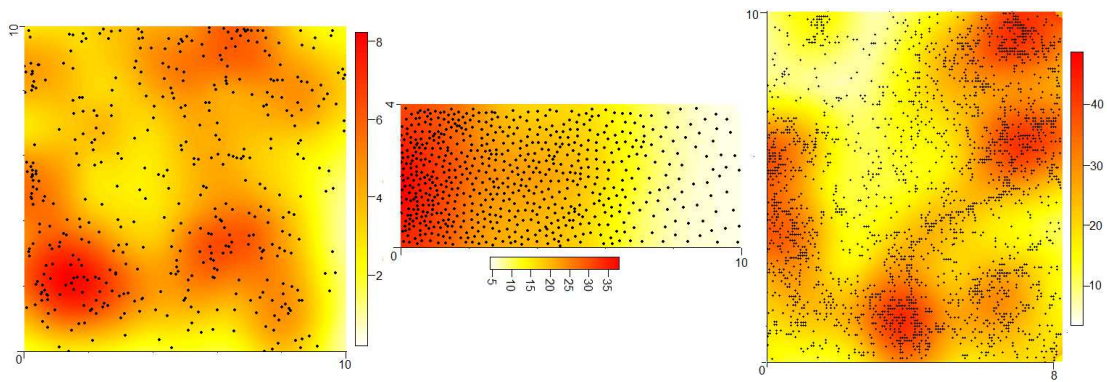}
		\caption{Examples of continuously varying estimated intensity function of Poisson processes. From left to right: white oaks in Lansing Wood, particles in a bronze filter and fires in New Brunswick.}
		\label{fig1}}
\end{figure}

A piecewise constant structure for the IF also allows for the analysis of the point pattern phenomenon to be performed under a cluster analysis perspective. This may be quite useful and interpretable in some applications. Each region with a constant IF constitutes a cluster and the clustering structure may be related to some practical aspect of the problem.

This paper considers a class of Cox process models with piecewise constant intensity function that is able to define the regions in which the IF is constant in a flexible way. The motivation is to have models that are suitable to explain and predict the variability of point process phenomena, yet providing more precise estimates than methodologies with continuously varying IF's. That is achieved through the \emph{level-set Cox process} (LSCP), originally proposed in \citet{hildeman2017level}, which is based on a structure proposed in \citet{dunlop2016hierarchical} to define a piecewise constant function in a given space by the levels of a latent Gaussian process (GP). This means that the region in which the intensity function assumes a given value is defined by the region in which a latent Gaussian processes assumes values in a given interval. This construction is considerably flexible to define space partitions, allowing for various shapes and sizes of the regions, including disjoint regions with the same IF.
\citet{hildeman2017level} actually proposes a more general version of level-set Cox process in which the observed point process follows independent log-Gaussian Cox processes in each region defined by the random partition, meaning that the IF depends on independent Gaussian process in each region. Therefore, the level-set Cox process considered in this paper is a particular case of the model proposed in \citet{hildeman2017level}, in which the IF is constant inside each region.

The methodology from \citet{hildeman2017level} however considers a discretized (finite-dimensional) approximation of the originally proposed model to approach the problem of performing statistical inference based on observations of the level-set Cox process. The authors argue that ``some finite-dimensional approximation of the LSCP model is needed if it is to be used for inference". The discrete approximation is based on a regular lattice that defines a joint model on the number of observations in each cell of the lattice as conditionally (on the respective rates) independent Poisson distributions. An important consequence of this approach, as it is mentioned by the authors, is that the information on the fine-scale behavior of the point pattern is lost. Furthermore, the latent Gaussian processes from the original LSCP are replaced by the respective multivariate normal distributions on one location inside each square of the lattice (usually the center). Finally, although the authors provide results to establish that the posterior distribution based on the discrete approximation converges (in total variation distance) to the posterior distribution under the continuous model, no bounds for the approximation error are provided.

The particular case of the LSCP proposed in \citet{hildeman2017level} in which the IF is piecewise constant with only two levels is proposed in \citet{Mylly}, where the authors also consider a discrete approximation of the process to perform Bayesian inference. This, in turn, is a special case of the random-set-generated Cox process described in \citet[p. 382]{Illian}, in which the random dynamics of the random partition (in 2 regions) is not specified.

The main aim of this paper is to devise an exact methodology to perform Bayesian inference for level-set Cox process models in which the IF is piecewise constant. The term exact here means that no discrete approximation of any kind is assumed and MCMC error is the only source of inaccuracy, as in any standard Bayesian analysis. This is not a trivial task due to: i) the intractability of the likelihood function of the proposed model (to be made clear in Section \ref{sec_model}); and ii) the infinite dimensionality of the model's parameter space due to the latent Gaussian process component. These two issues arise in several classes of statistical models nowadays \citep[see, for example,][]{bpr06a,G&G,GRL} and, given the high complexity involved, it is common to only find solutions in the literature that are based on discretization of continuous processes, as is the case with LSCP. The use of such approximations however has considerable disadvantages \citep[see][]{simpson2016going}. It induces a bias in the estimates, which is typically hard to quantify and control. Furthermore, even if limiting results guarantee some type of convergence to the continuous model when the discretization gets finer, the computational cost involved to get reasonably good approximations may be unknown and/or too high. Finally, discrete approximations may lead to serious model mischaracterization, compromising the desired properties of the model.

The exact inference methodology proposed in this paper makes use of a simulation technique called retrospective sampling which basically allows to deal with infinite-dimensional random variables by unveiling only a finite-dimensional representation of this. A pseudo-marginal MCMC algorithm that converges to the exact posterior distribution of all the unknown quantities in the model is proposed. These quantities include the intensity function and the random partition that defines the piecewise constant structure. Also, the Monte Carlo approach makes it straightforward to sample from the posterior predictive distribution of various appealing functions.

The known high computational cost involved in algorithms that deal with the simulation of Gaussian processes is mitigated by considering the nearest neighbor Gaussian process (NNGP) \citep{NNGP}. This has a particular conditional independence structure that leads to a sparse covariance structure and, consequently, to huge computational gains when compared to traditional Gaussian processes. An example with more than 5 thousand observations is presented in Section \ref{sec_app}. This is, to the best of our knowledge, the first work to consider a latent NNGP within a complicated likelihood structure that does not allow for directly sampling from the posterior or full conditional distribution of the NNGP component. In this sense, this paper also offers some methodological contributions to deal with latent NNGPs.

Another major contribution in this paper is the introduction of an extension of LSCP to consider spatiotemporal processes. This means that the point process is observed on the same region over discrete time and a temporal correlation structure is introduced in the model to explain the evolution of both the space partition and the IF levels.

Section \ref{sec_model} of the paper presents the level-set Cox process model and discusses its most important properties. The proposed MCMC algorithm and the extension to consider spatiotemporal models is presented in Section \ref{sec_inf}, which also addresses some relevant computational issues. Section \ref{sec_sim} explores some simulated examples to discuss important aspects and investigate the efficiency of the proposed methodology. Finally, Section \ref{sec_app} applies the methodology to some real datasets. For one of them, results are compared to those obtained with a continuously varying IF Cox process model.

\section{Level-set Cox process models}\label{sec_model}

Let $Y=\{Y(s): s \in S\}$ be a Poisson process in some compact region $S\in\mathbb{R}^n$ with intensity function $\lambda_{S}=\{\lambda(s): s \in S\}$, $\lambda(s): S  \rightarrow \mathbb{R}^+$ and define $\mathbf{S}_K=\{S_1,\ldots,S_K\}$, $K\in \mathbb{N}$, to be a finite partition of $S$. We shall focus on the case where $S\subset\mathbb{R}^2$ given its practical appealing, although all the definitions and results to be presented in this paper are valid for $\mathbb{R}^n$ or any other measurable space \citep[see][]{kingman1993poisson}. Now let $\lambda=(\lambda_1, \ldots, \lambda_K)$ be a vector of positive parameters such that the IF of $Y$ on $S_k$ is $\lambda_k$, $k=1,\ldots,K$. Let also $c=(c_1,\ldots, c_{K-1}) \in \mathbb{R}^{K-1}$, $-\infty=c_0 < c_1 < \ldots < c_{K-1}< c_K=\infty$, be the values that define the level sets of a latent Gaussian process $\beta$ on $S$ and, consequently, the finite partition $\mathbf{S}_K$ of $S$. The level-set Cox process model is then defined as follows:
\begin{eqnarray}
\ds (Y|\lambda_{S}) &\sim& PP(\lambda_{S}),   \\
\lambda(s) &=& \ds\sum_{k=1}^{K} \lambda_k I_{k}(s),\;s\in S,  \\
S_k &=& \{ s\ \in\ S;\; c_{k-1} < \beta(s) < c_k    \},\; k=1,\ldots,K, \\
\beta &\sim& GP(\mu,\Sigma(\sigma^2,\tau^2)),\\
\pi(c) &=& \mathds{1}(c_1<\ldots<c_{K-1}),\\
\lambda &\sim& prior,
\end{eqnarray}
where $I_{k}(s)$ is the indicator of $\{s \in S_k\}$ and $GP(\mu,\Sigma(\sigma^2,\tau^2))$ is a stationary Gaussian process with mean $\mu$ and covariance function $\Sigma(\sigma^2,\tau^2)$, where $\sigma^2$ is the stationary variance and $\tau^2$ is a range parameter indexing the correlation function. The prior on $\lambda$ will be properly defined in Section \ref{subseccov}. Note that the levels of the IF are not a function of the GP, which simply specifies, along with $c$, the partition of $S$ that defines the piecewise constant structure. Finally, $\beta$, $c$ and $\lambda$'s are assumed to be independent \emph{a priori}. One may also consider an IF of the type $\lambda(s) = \sum_{k=1}^{K} \kappa(s)\lambda_k I_{k}(s)$, where $\kappa(s)$ is a known offset term. This is useful for example when observing cases of some human disease in in a region with varying population density.

Notice that the likelihood of the proposed level-set Cox process model is not identifiable. That is because, for each point in the (infinite-dimensional) parameter space, there are an uncountable number of other points that return the same likelihood value. That is basically implied by the non-identification of the scale of the GP $\beta$. In order to see that, let us redefine $\beta$ as $\beta=\mu + \sigma \beta^{*}$, where $\beta^{*} \sim N(0,\Sigma(1,\tau^2))$. Then, any transformation of the type $\mu^{*}=a\mu + b$, $\sigma^{*}=a\sigma$ and $c^{*}_{k}= b + a c_{k}$, $\forall m \in \mathbb{R},\; a \in \mathbb{R}^+$, $\forall k$, defines the same partition $\mathbf{S}_K$ and, consequently, the same likelihood value. A simple way to solve this problem whilst not compromising the flexibility of the model is to fix either $c$ or the hyperparameters $(\mu,\sigma^2)$. We shall adopt the latter, which also avoids the high complexity involved in estimating those parameters.

Model identifiability could also be compromised, in theory, by label-switching of the coordinates of $\lambda$. Nevertheless, given the complexity of the sample space, this is not expected to happen in an MCMC context, as it was the case for all the examples to be presented in this paper.

A theoretical limitation of the model is the neighboring structure implied by the continuity of the latent Gaussian process. For any model with $K\geq3$, regions 1 and K share a border with only one other region, 2 and $K-1$, respectively, and any other region $k$ shares a border with regions $k-1$ and $k+1$. Whilst this represents a clear theoretical limitation of the model, it is not expected to be a practical problem in most cases. That is because the uncertainty around the borders is higher than the uncertainty away from them, so the need to pass through a third region to change between two other ones should typically not affect the model fitting. Furthermore, the estimated ordering of the $\lambda_k$'s will consider the likelihood of the different neighboring configurations. Figure \ref{fig2} shows an example of a neighboring structure with $K=3$ that is not contemplated by the proposed model and three possible structures that may be estimated. Despite this restriction, we highlight the great flexibility of the model to define the partition $\mathbf{S}_K$ of $S$. Basically, given the neighboring restriction described above, any smooth partition of the space is contemplated by the model. In particular, it is possible to have disjoint regions with the same IF.

\begin{figure}[h!]
	\centering{
		\includegraphics[width=0.85\linewidth]{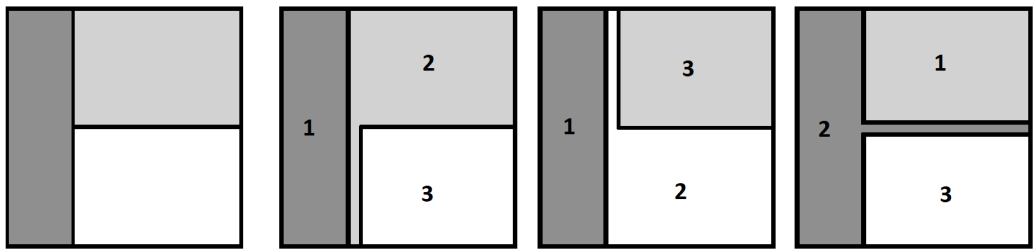}
		\caption{Example of a neighboring structure with $K=3$ that is not contemplated by the proposed model (far left) and three possible structures that may be estimated.}
		\label{fig2}}
\end{figure}

We consider the number of levels $K$ for the IF to be fixed. The choice of this value may be based on prior information about the phenomenon, the type of structure the researcher wants to estimate, or even some empirical analysis of the data, for example, based on kernel smoothing estimates of the IF \citep[see][]{rowlingson2012package}. The choice of $K$ should consider the trade-off between fitting and parsimony and take into account the scale of the Poisson distribution.

Finally, the piecewise constant structure allows for the analysis to be performed under a cluster analysis perspective. This may be useful and interpretable in some applications. Each region $S_k$ constitutes a cluster and the clustering structure may be related to some practical aspect of the problem.

\subsection{Nearest neighbor Gaussian process prior for $\beta$}

The computational bottleneck of the methodology proposed in this paper is the sampling of the Gaussian process $\beta$, which is performed not only when updating this component but also, retrospectively, when updating $N^*$ and $\lambda$. The cost to simulate from a $d$-dimensional multivariate normal distribution is $\mathcal{O}(d^3)$ and, in our case, $d$ will typically be in the order of $10^3$ or $10^4$. Several solutions have been proposed in the literature to deal with this problem. Some of them are exact in the sense that the approximation process defines a valid probability measure. This property is highly desirable as it guarantees that the analysis is performed under the Bayesian paradigm. In this work, we consider the use of the Nearest neighbor Gaussian process (NNGP), proposed in \citet{NNGP}. The NNGP process was originally designed to approximate some Gaussian process, called the parent GP, in classical geostatistical problems in which the (discretely) observed process is either the GP itself or the GP + i.i.d. noise. In our context, the GP is latent in a more complex way. Nevertheless, it is used only to determine the partition of $S$ and not the actual values of the IF. For that reason and because the NNGP defines a valid Gaussian process probability measure, it is reasonable to see the NGPP simply as the GP prior for $\beta$ and not an approximation for desirable traditional GP.

The NNGP is a valid Gaussian process, devised from a parent $GP(\mu,\Sigma(\sigma^2,\tau^2))$ by imposing some conditional independence structure that leads to a sparse covariance structure. For a reference set $\mathcal{S}=\{\mathfrak{s}_1,\ldots,\mathfrak{s}_r\}$ and a maximum number $m$ of neighbors, the NNGP factorizes the distribution of $\beta$ (conditional on parameters) as follows:
\begin{eqnarray}
\ds \pi(\beta) &=& \pi(\beta_{\mathcal{S}})\pi(\beta_{S\setminus\mathcal{S}}|\beta_{\mathcal{S}}),   \\
\pi(\beta_{\mathcal{S}}) &=& \pi_{PG}(\beta_{\mathfrak{s}_1})\pi_{PG}(\beta_{\mathfrak{s}_2}|\beta_{\mathfrak{s}_1})\pi_{PG}(\beta_{\mathfrak{s}_3}|\beta_{\mathfrak{s}_1},\beta_{\mathfrak{s}_2})\ldots
\pi_{PG}(\beta_{\mathfrak{s}_{m+1}}|\beta_{\mathfrak{s}_1},\ldots,\beta_{\mathfrak{s}_m})\\
&&\pi_{PG}(\beta_{\mathfrak{s}_{m+2}}|\beta_{\mathcal{N}(\mathfrak{s}_{m+2})})\ldots
\pi_{PG}(\beta_{\mathfrak{s}_{r}}|\beta_{\mathcal{N}(\mathfrak{s}_{r})}),  \\
\pi_{PG}(\beta_{S_0}|\beta_{\mathcal{S}}) &=& \prod_{i=1}^I\pi_{PG}(\beta{s_i}|\beta_{\mathcal{N}(s_{i})}),\;\mbox{for any finite set}\;S_0=\{s_1,\ldots,s_I\}\subset S\setminus\mathcal{S},
\end{eqnarray}
where $\pi_{PG}$ is the respective density under the parent GP measure, $\mathcal{N}(\mathfrak{s}_{i})$ is the set of the $m$ closest neighbors of $\mathfrak{s}_{i}$ in $\{\mathfrak{s}_{1},\ldots,\mathfrak{s}_{i-1}\}$, for $i\geq m+2$, and $\mathcal{N}(s_{i})$ is the set of the $m$ closest neighbors of $s_{i}$ in $\mathcal{S}$. We shall refer to the resulting NNGP process as $NNGP(\mu,\tilde{\Sigma}(\sigma^2,\tau^2))$.

In traditional geostatistical models, in which the GP is observed (with error) in a missing completely at random (MCAR) set of locations, the reference set is conveniently defined to be the locations of the observations. In our context however, by the very nature of the process being observed, that is not a reasonable choice. Instead, we set $\mathcal{S}$ to be a regular lattice on $S$. Based on the results in \citet{NNGP} and results of several simulated examples with our model, we set $r=2500$ and $m=16$.

The NNGP leads to massive gains in computational cost when compared to traditional GPs due to its particular conditional independence structure. Moreover, the distribution of $\pi(\beta_{S_0}|\beta_{\mathcal{S}})$ is conditionally independent among the locations in any finite set $S_0\subset S\setminus\mathcal{S}$, which means that the algorithm to sample from this distribution can be parallelized. This is an appealing feature in our case as all the locations from $Y$ and $N$ are in $S\setminus\mathcal{S}$ and have will be sampled from the NNGP prior on every iteration of the MCMC algorithm. The specific steps of the proposed MCMC algorithm that can be parallelized are indicated in the algorithm shown in Appendix B.

\section{Bayesian inference}\label{sec_inf}

Inference for the level-set Cox process model is performed under the Bayesian paradigm, meaning that it is based on the posterior distribution of all the unknown quantities of the model.
As it was mentioned before, the stationary mean and variance of the Gaussian process are fixed to identify the model. We also choose to fix the correlation parameter $\tau^2$. We believe this can be done in a reasonable way based on the scale of the domain $S$ - this issue will be discussed in Section \ref{subseccov} and explored in the simulated studies in Section \ref{subsecsens}. Also, fixing the parameters that index the GP brings huge computational gains to the inference process. Regarding parameter $c$, we set a uniform improper prior with the restriction $c_1<\ldots<c_{K-1}$.

Defining $\theta=\{\lambda,c,\beta\}$ to be all the unknown quantities of the model, the likelihood function for the level-set Cox process model is obtained by writing the density of $Y$ w.r.t. the measure of a unit intensity Poisson process on $S$ \citep[see][]{GF19}, which is given by
\begin{eqnarray}\label{eq_lik}
\ds L(\theta;Y) \propto \exp  \left\{ -\displaystyle \sum_{k=1}^{K} \lambda_k \mu_k \right\} \displaystyle \prod_{k=1}^{K}  \left(\lambda_k \right)^{|Y_k|} ,
\end{eqnarray}
where $\mu_k$ is the area of $S_k$ and $|Y_k|$ is the number of events from $Y$ on $S_k$.

The posterior distribution of $\theta$ as well as the full conditional distribution of any of its components have densities proportional to the joint density $\pi(\theta,Y)=L(\theta,Y)\pi(\theta)$, such that
\begin{equation}
\ds \pi(\theta,Y) \propto  \exp  \left\{ -\displaystyle \sum_{k=1}^{K} \lambda_k \mu_k \right\} \left[\displaystyle \prod_{k=1}^{K}  \left(\lambda_k \right)^{|Y_k|}\pi(\lambda_k)\right]\left[\displaystyle \prod_{k=1}^{K-1} \pi(c_k) I_{\scriptstyle(c_1<\ldots<c_{k-1})}\right]\pi_{PG} (\beta), \label{eq_jointdens}
\end{equation}
where $\pi_{PG} (\beta)$ is written w.r.t. some suitable dominating measure, which is irrelevant for the derivation of the inference methodology.

Given its complexity, the posterior distribution of $\theta$ is assessed via MCMC. This is not a trivial task for two main reasons. First, the MCMC algorithm is infinite-dimensional because of the infinite dimensionality of the coordinate $\beta$. Second, the likelihood in (\ref{eq_lik}) is analytically intractable, since the areas $\mu_k$ of the regions $S_k$ cannot be computed exactly. It is then quite challenging to devise a valid and efficient MCMC that is exact in the sense of converging to the exact posterior distribution of $\theta$.

To deal with the infinite-dimensionality of $\beta$ we resort to a simulation technique called retrospective sampling. In the context of simulation of infinite-dimensional random variables, this means that only a finite-dimensional representation of the infinite-dimensional r.v. is simulated and this representation has the following two properties: i) it is enough to unveil only this representation to execute the algorithm in context (an MCMC in our case); ii) any finite-dimensional part of the infinite-dimensional remainder of that r.v. can be simulated conditional on this representation. This means that the GP $\beta$ is to be simulated only at a finite (but random) collection of locations on each iteration of the MCMC chain.
It is the particular random structure of those locations that guarantee the two properties above. The idea of retrospective sampling in the context of simulation of infinite-dimensional r.v.'s was introduced  in \citet{EA0} to perform exact simulation of diffusion paths. It was later used in a statistical context in several works \citep[see, for example,][]{bpr06a,G&G,GRL}.

The intractability of the likelihood function precludes us from performing standard Metropolis-Hastings (MH) steps for any coordinate of the chain since all of them appear in this function. Our solution resorts to a powerful and flexible general MCMC algorithm called the pseudo-marginal Metropolis-Hastings (PMMH), proposed in \citet{andrieu2009}. This algorithm allows us to replace the likelihood terms in the expression of the MH acceptance probability with a pointwise unbiased and almost surely non-negative estimator of the likelihood function. This leads to an augmented Markov chain that has the desired posterior distribution as the marginal invariant distribution of the chain - marginalized w.r.t. the random seed of the aforementioned unbiased estimator. Naturally, the efficiency of the algorithm relies on the properties of that estimator. Roughly speaking, the smaller is its variance the better \citep[see][]{vihola}.

The only unknown quantities in the expression of the likelihood function in (\ref{eq_lik}) are the areas $\mu_k$. This means that, in order to obtain an unbiased estimator for the likelihood, we need an unbiased estimator for $M=\exp  \left\{ -\displaystyle \sum_{k=1}^{K} \lambda_k \mu_k \right\}$. Note that this quantity does not depend on the observed Poisson process events. Although unbiased estimators for the $\mu_k$'s can be easily obtained using uniform r.v.'s on $S$, it is not straightforward to devise an unbiased estimator for $M$ from this. We resort to a neat class of unbiased estimators called the Poisson estimator \citep[see][]{bpr06a}, which devises an unbiased estimator for $M$ as a function of a random Poisson number of uniformly distributed r.v.'s on $S$. The unbiased estimator for $M$ and some of its important properties are given in Propositions \ref{prop1} and \ref{prop2}, respectively.

\begin{proposition}\label{prop1}
Let $N$ be a unit rate Poisson process in the cylinder with base $S$ and height in $[0,+\infty)$ and define $N=g(N^{*},\lambda^*)$ as the projection on $S$ of the points from $N^*$ that have height smaller than $\lambda^*$, where $\lambda^*= (\delta\lambda_M -\lambda_m)$, $\lambda_{M}=\underset{k}{\max}\{\lambda_k\}$ and $\lambda_{m}=\underset{k}{\min}\{\lambda_k\}$. Then, for any $\delta>1$, an unbiased and almost surely positive estimator for $M$ is given by
\begin{eqnarray}
\ds \hat{M}=e^{-\mu(S)\lambda_{m}} \displaystyle\displaystyle\prod_{k=1}^{K} \left( \frac{\delta\lambda_{M} - \lambda_k}{\delta\lambda_{M} - \lambda_{m}}\right)^{|N_k|}, \label{PEM}
\end{eqnarray}
where $\mu(S)$ is the area of $S$ and $|N_k|$ is the number of points from $N$ falling in $S_k$.
\end{proposition}
\begin{proposition}\label{prop2}
Estimator $\hat{M}$ has a finite variance which is a decreasing function of $\delta$.
\end{proposition}
\begin{proof}
See Appendix A for the proofs of Propositions \ref{prop1} and \ref{prop2}.
\end{proof}

In our retrospective sampling context, it is $N$ that determines the locations at which $\beta$ is to be simulated, besides the locations from $Y$. Furthermore, the mean number of locations from $N$ is $(\delta\lambda_M -\lambda_m)\mu(S)$, which gives the intuition for the result in Proposition \ref{prop2}. This establishes a trade-off related to the choice of $\delta$, as an increase in its value reduces the variance of $\hat{M}$ (and consequently improves the mixing of the MCMC chain) but increases the computational cost per iteration of the MCMC algorithm (and vice-versa).

We define a pseudo-marginal MCMC algorithm to sample from the posterior distribution of $\theta$ based on the estimator in (\ref{PEM}). On each iteration of the Markov chain, the general algorithm proposes a move $(\theta,N^*)\rightarrow(\ddot{\theta},\ddot{N}^*)$ from a density $q(\ddot{\theta},\ddot{N}^*|\theta,N^*)=q(\ddot{\theta}|\theta)q(\ddot{N}^*|N^*)$, where $q(\ddot{N}^*|N^*)=q(\ddot{N}^*)$ is the Poisson process defined in Proposition \ref{prop1}, which we shall call the pseudo-marginal proposal, and accepts with probability given by
\begin{equation}\label{jdens}
\ds 1 \wedge \left(\frac{\hat{\pi}(\ddot{\theta};\ddot{N^*})}{\hat{\pi}(\theta;N^*)}\frac{q(\theta|\ddot{\theta})}{q(\ddot{\theta}|\theta)}\right),
\end{equation}
where
\begin{equation}\label{a.p.2}
\ds \hat{\pi}(\theta;N^*) =  e^{-\mu(S)\lambda_{m}} \left[\displaystyle\prod_{k=1}^{K} \left( \frac{\delta\lambda_{M} - \lambda_k}{\delta\lambda_{M} - \lambda_{m}}\right)^{|N_k|}  \left(\lambda_k \right)^{|Y_k|}\pi(\lambda_k)\right] \pi(c) \pi_{PG}(\beta).
\end{equation}

The algorithm above is bound to be inefficient as it is, given the complexity of the chain's coordinates. We adopt simple yet important changes to obtain a reasonably efficient algorithm. First, we split the coordinates into blocks, making this a Gibbs sampling with pseudo-marginal MH steps. This implies that the acceptance probability of any block is also given by (\ref{jdens}). Also, the choice to define $N$ as a function of $N^*$, with the distribution of the latter being independent of $\theta$, instead of working directly with $N$, allows us to sample $N^*$ alone as one block of the Gibbs sampler. Furthermore, note that $N^*$ has an infinite collection of points but, in order to evaluate the acceptance probability in (\ref{jdens}), it is enough to unveil $N$ - the projection on $S$ of the points of $N^*$ that are below $\lambda^*$. This means that, as is the case for $\beta$, $N^*$ is also sampled retrospectively. The blocks of the Gibbs sampling are: $N^*$, $\beta$, $\lambda$, $c$. The algorithm to sample from each block is described below.

\subsection{Sampling $N^*$}

The standard version of the pseudo-marginal algorithm proposes a move in $N^*$ from the pseudo-marginal proposal $q(N^*)$ and accepts with probability given by (\ref{jdens}). Furthermore, the fact that the acceptance probability in (\ref{jdens}) depends on $N^*$ only through the points falling below $\lambda^*$, implies that we only need to unveil $N$ in order to update $N^*$. Nevertheless, since $N^*$ will typically have many points falling below $\lambda^*$, this proposal might have a low acceptance rate which, in turn, may compromise the mixing of the chain.

Instead, we adopt a proposal distribution that updates $N^*$ below and above $\lambda^*$, separately. The latter is proposed from the pseudo-marginal proposal and accepted with probability 1, given that it does not appear in (\ref{jdens}). For that reason, this step is only performed conceptually and points from $N^*$ above $\lambda^*$ are sampled retrospectively, if required (when the proposal value for $\lambda$ leads to a higher value of $\lambda^*$), from a $PP(1)$.

For the points below $\lambda^*$, we split $S$ into $L$ regular squares (assuming that $S$ is a rectangle) and update $N^*$ in each of the respective cylinders separately. Standard properties of Poisson processes imply that, under the pseudo-marginal proposal, $N$ is mutually independent among the $L$ cylinders and follows a $PP(\lambda^*)$ in each of them. This splitting strategy imposes an optimal scaling problem w.r.t. $L$. Empirical analyzes for several simulated examples (some of which are presented in Section \ref{sec_sim}) suggest that the value of $L$ should be chosen so that the average acceptance rate among the $L$ squares is around 0.8.

We shall refer to $N$ restricted to the $l$-th square as $N_l$. A move $N_l\rightarrow \ddot{N}_l$ is accepted with probability
\begin{equation}\label{DTPMN2}
\ds \alpha_{N_l}=  \left(\prod_{k=1}^{K} \frac{\delta\lambda_{M} - \lambda_k}{\delta\lambda_{M} - \lambda_{m}}\right)^{|\ddot{N}_l|-|N_l|},
\end{equation}
where $|N_l|$ is the number of points from $N_l$.

\subsection{Sampling $\beta$}

The latent process $\beta$ is sampled retrospectively, due to its infinite dimensionality. This means that it is sampled at a finite collection of locations which are enough to perform all the steps of the MCMC algorithm. This collection is defined by the locations of $\mathcal{S}$, $Y$ and $N$, with the third one changing along the MCMC on the update steps of $N^*$ and $\lambda$.

The first important fact to be noted here is that we are unable to sample $\beta$ directly from its full conditional distribution. Second, the proposal for $\beta$ has to be such that the expression of the acceptance probability in (\ref{jdens}) can be analytically computed. More specifically, we need a proposal distribution that cancels out the term $\pi_{GP}(\beta)$.

The conditional independence structure of the NNGP demands extra care to specify its proposal distribution. For example, it is unwise to define a proposal that, at each iteration of the MCMC, fixes $\beta$ at a random finite collection of points from $S$ and propose the remainder from the NNGP prior. Conditional distributions under the NNGP that do not follow the ordering in $\mathcal{S}$ will not benefit from the conditional independence structure to have computational gains.

A adopt a non-centered random walk proposal for $\beta$. More specifically, a move $\beta\rightarrow\ddot{\beta}$ is proposed from:
\begin{eqnarray}\label{mod_rw_b}
\ddot{\beta}(s) &=& \sqrt{1-\varsigma^2}\beta(s)+\varsigma\varepsilon(s),\;\;s\in S, \\
\varepsilon &\sim& NNGP(0,\tilde{\Sigma}). \nonumber
\end{eqnarray}
This proposal is called the preconditioned Crank–Nicolson proposal (pCN) and was introduced by \citet{cotter}, not in an NNGP context. In a finite-dimensional context, the pCN proposal differs slightly from the traditional centered random walk but, unlike the latter, leads to an acceptance probability that does not depend on the prior density of the component being updated. Furthermore, the pCN proposal is valid also in the infinite-dimensional context, as defined in (\ref{mod_rw_b}), whereas the centered random walk is not. The proposal variance $\varsigma^2$ is chosen so to have an acceptance rate of approximately 0.234 \citep[see][]{cotter}.

The acceptance probability of a move $\beta\rightarrow\ddot{\beta}$ is given by
\begin{equation}\label{a.p.beta}
\ds \alpha_{\beta}=1 \wedge
\left(\ds\prod_{k=1}^{K} \left( \frac{\delta\lambda_{M} - \lambda_k}{\delta\lambda_{M} - \lambda_{m}}\right)^{|\ddot{N}_k|-|N_k|} \left(\lambda_k \right)^{|\ddot{Y}_k|-|Y_k|}\right),
\end{equation}
where $|N_k|$ and $|Y_k|$ are the respective values obtained from $\beta$ and $|\ddot{N}_k|$ and $|\ddot{Y}_k|$ are the respective values obtained from $\ddot{\beta}$.

\subsection{Sampling $\lambda$ and $c$}

The vector $\lambda$ is sampled jointly from a proposal given by a Gaussian random walk with a properly tuned covariance matrix that is adapted, based on the respective empirical covariance matrix of the chain, up to a certain iteration \citep[see][]{roberts2009examples} so to have the desired acceptance rate - varying from 0.4 to 0.234 according to the dimension of $\lambda$. The acceptance probability of a move $\lambda\rightarrow\ddot{\lambda}$ is given by\\
\begin{equation}\label{a.p.lambda}
\ds \alpha_{\lambda}=1 \wedge \left(e^{-\mu(S)(\ddot{\lambda}_m-\lambda_m)} \left[\displaystyle\prod_{k=1}^{K}
\frac{\left( \frac{\delta\ddot{\lambda}_M - \ddot{\lambda}_k}{\delta\ddot{\lambda}_M - \ddot{\lambda}_m}\right)^{|\ddot{N}_k|}}
{\left( \frac{\delta\lambda_M -\lambda_k}{\delta\lambda_M - \lambda_m}\right)^{|N_k|}} \left(\frac{\ddot{\lambda}_k}{\lambda_k}\right)^{|Y_k|}\right]\frac{\pi(\ddot{\lambda})}{\pi(\lambda)} \right),
\end{equation}
where $\pi(\lambda)$ is the prior density of $\lambda$ to be defined in Section \ref{subseccov}. Also, $\ddot{N}_k$ is the respective value obtained from $\ddot{N}=g(N^{*},\ddot{\lambda}^*)$ and $\ddot{\lambda}^*=(\delta\ddot{\lambda}_M-\ddot{\lambda}_m)$.

The parameter vector $c$ is jointly sampled from a uniform random walk proposal with a common (and properly tuned) length for each of its components. If the ordering of the proposed values is preserved, a move $c\rightarrow\ddot{c}$ is accepted with probability
\begin{equation}\label{a.p.c}
\ds \alpha_{c}=1 \wedge
\left(
\displaystyle\prod_{k=1}^{K} \left( \frac{\delta\lambda_{M} - \lambda_k}{\delta\lambda_{M} - \lambda_{m}}\right)^{|\ddot{N}_k|-|N_k|} \left(\lambda_k \right)^{|\ddot{Y}_k|-|Y_k|}
\right),
\end{equation}
where $\ddot{N}_k$ is the respective value obtained from the $S_k$ region defined by $\ddot{c}$.

\subsection{Computational aspects}

\subsubsection{Covariance function and model identifiability}\label{subseccov}

The specification of the covariance function $\Sigma(\sigma^2,\tau^2)$ of the parent process in the NNGP prior plays an important role in the proposed methodology. Empirical results of several simulated examples (omitted here) suggest that the powered exponential with an exponent close to 2 is a good and robust choice. This is given by

\begin{equation}\label{a.p.lambda}
\ds Cov(\beta(s),\beta(s'))=\exp\left\{ -\frac{1}{2\tau^2}|s-s'|^{\gamma} \right\},
\end{equation}
where $|s-s'|$ is the Euclidian distance between $s$ and $s'$ and we set $\gamma=1.95$. The specification of the parameter $\tau^2$ is related to the smoothness of the estimated IF and to model identifiability issues as discussed next.

We call the reader's attention to the fact that the Poisson process likelihood in (\ref{eq_lik}) is ill-posed. Note that the likelihood function increases indefinitely as the IF increases in balls centered around the observations, with these balls getting smaller, and approaches zero outside them. The Cox process formulation is a way to regularize the likelihood function by assigning a prior to the IF which is, in our particular level-set formulation, a non-parametric prior. Naturally, this prior has a great impact on the resulting posterior distribution. In particular, Bayes theorem implies that the posterior distribution of $\beta$ is absolutely continuous w.r.t. its prior, implying that all the almost surely properties under the prior are preserved under the posterior. Due to the aforementioned misbehavior pattern of the likelihood function, the information contained in the data about the likelihood will favor values of the IF that go in the direction of the characteristics described at the beginning of this paragraph. As a consequence, the likelihood will favor smaller values of the smoothness parameter $\tau^2$ - that make the IF less smooth. For that reason, fixing the value of $\tau^2$ is a reasonable strategy. The value of this parameter will determine the smoothness of the estimated IF and should therefore be chosen based on the researcher's preference. We believe that, typically, partitions with very small regions should be avoided. In such cases, it might be more reasonable to resort to continuously varying intensity functions. Another reason why the estimation of $\tau^2$ should be avoided is the fact that the full conditional distribution of this parameter would depend on the joint density of the NNGP at all the locations of $Y$ and $N$ and, as the locations of $N$ are latent (missing data) and numerous, the mixing of the MCMC chain could be seriously compromised. The choice of $\tau^2$ is discussed and illustrated in the simulated examples in Section \ref{subsecsens}.

Model identifiability issues may arise from the existence of local modes in the posterior density, especially for small datasets, and the prior information on the IF may not be enough to avoid the existence of significant local modes. A reasonable way to mitigate that is by adding extra coherent prior information in the model through the prior distribution of $\lambda$. Under a model parsimony perspective, it is reasonable to fit LSCPs with fewer levels and clearly distinct rate values than with more levels with similar rate values. One way to introduce this information in the model and, consequently, improve model identifiability, is by adopting a joint repulsive prior for $\lambda$ such that the $\lambda_k$'s tend to repel each other. We define a repulse prior based on the $Rep$ distribution proposed in \citet{quintana20}. However, instead of directly penalizing the differences between the $\lambda_k$'s, we consider a scaled version of those differences, as follows.
\begin{eqnarray}\label{prior_lambda}
\pi(\lambda) &\propto& \left[\prod_{i=1}^K\pi_G(\lambda_k)\right]R(\lambda;\rho,\nu), \nonumber\\
\pi_G(\lambda_k) &\propto& \lambda_{k}^{\alpha_k-1}e^{-\eta_k\lambda_k},\;\alpha_k>0,\;\eta_k>0,\;k=1,\ldots,K, \nonumber \\
R(\lambda;\rho,\nu) &=& \prod_{1\leq k_1<k_2\leq K}\left(1-\exp\left\{-\rho\left(\frac{|\lambda_{k_1}-\lambda_{k_2}|}{\sqrt{\lambda_{k_1}+\lambda_{k_2}}}\right)^{\nu}\right\}\right).
\end{eqnarray}
We shall call this the repulsive gamma prior with notation $RG(\alpha,\eta,\rho,\nu)$.
The scale factor $\sqrt{\lambda_{k_1}+\lambda_{k_2}}$ is meant to penalize the proximity of the $\lambda_k$'s considering the scale of the Poisson distribution. For example, it is not reasonable to equally penalize the pairs $(5,2)$ and $(13,10)$. Note that the scale fator is the sum of the standard deviations of the Poisson distributions with means $\lambda_{k_1}$ and $\lambda_{k_2}$. Results from simulated studies with different combinations suggest that $\rho\in[1,5]$ and $\nu=3$ are reasonable choices. The plot of the penalizing factor $r(x)=\left(1-\exp\left\{-\rho x^{\nu}\right\}\right)$ is shown in Figure \ref{figap1} in Appendix C. Note that the repulsive prior is proper since $\pi_G$ is a probability density and $R(\lambda;\rho,\nu)$ is bounded.

The RG prior on $\lambda$ may be useful to identify the most suitable value of $K$ to be used. This value is typically chosen based on an empirical analysis of the kernel smoothing estimates by analyzing aspects of the estimated IF such as minimum and maximum values, homogeneity and estimated value in regions of the space domain, variation across the spatial domain. These aspects ought to be interpreted in terms of the standard deviation of the Poisson distribution. Naturally, this is an empirical strategy and it might be wise to fit the model for different values of $K$. In this case, the RG prior can be very useful to indicate if a chosen value of $K$ is higher than necessary. As the prior repulses the values of the $\lambda_K$'s, the area of one or more regions in the partition may be estimated to be zero, effectively meaning that $K$ should be smaller. This happened to all the real examples analyzed in Section \ref{sec_app_spt}. Finally, model selection criteria can also be used to choose $K$, as it is shown in Section \ref{model_fit}.

\subsubsection{MCMC virtual updates}

Despite the NNGP prior, the computational cost of the MCMC algorithm may still be compromised by a large accumulation of points from $\beta$ resulting from successive rejections on the update step of this component and the simulation of extra points on the update steps of $\lambda$ and $N^*$.

The infinite dimensionality of $\beta$ and the retrospective sampling context provide an elegant and efficient solution for this problem. We add virtual update steps to the MCMC algorithm that update $\beta$ in $S\setminus\{\mathcal{S},Y,N\}$. Since the acceptance probability (\ref{jdens}) of the pseudo-marginal algorithm does not depend on $\beta$ at those locations, the proposal is accepted with probability 1. Furthermore, the retrospective sampling approach implies that those steps consist of simply deleting all the stored values of $\beta$ at $S\setminus\{\mathcal{S},Y,N\}$, justifying the term ``virtual" step. A virtual update is performed in-between every block update of the Gibbs sampling as long as the set of sampled locations of $\beta$ in $S\setminus\{\mathcal{S},Y,N\}$ is not empty at that moment of the algorithm.

\subsubsection{Other important issues}\label{susec_issues}

The choice of the initial values of the MCMC algorithm plays an important role to determine the efficiency of the algorithm in terms of mixing and estimation. Results from simulation studies suggest that it is reasonable to generate the initial values of $\beta$ from its NNGP prior and set $\lambda_{k}=|Y|/\mu(S)$, for all $k$. Typically, the ordering of the $\lambda_k$ parameters assumed in the first iterations of the MCMC does not change along the chain. This ordering will depend on the initial value of the GP and may not be the best one. Therefore, it may be convenient to impose a fixed ordering to the $\lambda_k$ parameters. Conditional on the initial values of $\lambda$, the initial value of $N$ is generated from pseudo-marginal proposal $q(N^*)$.

The choice of $\delta$ is also investigated in simulation studies under the trade-off defined by the fact that an increase in $\delta$ improves the mixing of the MCMC algorithm but also increases the computational cost per iteration of the algorithm. Results indicate that the distribution of the pseudo-marginal estimator has a very heavy tail to the right and that a choice of $\delta$ based on the mean number of points of $N$ is reasonably robust among different models and datasets. In particular, results were good for a variety of examples (some of which are presented in Section \ref{subsecsens}) when the mean number of points from $N$, under the pseudo-marginal distribution, was around 6000. This is a valid result for any scale and shape of the domain $S$. Note that both the function $M$ to be estimated by the pseudo-marginal estimator and $E[|N|]$ (under the pseudo-marginal distribution) depend on the areas of the partition regions only through the mean number of events in each region. Also, the variance of $\hat{M}$ depends on those areas through the mean number of events in each one and the relative differences $(\delta\lambda_M-\lambda_k)/(\delta\lambda_M-\lambda_m)$.

The proposed MCMC algorithm is highly parallelizable due to the conditional independence properties of the NNGP prior on $\beta$. In particular, the following two very expensive steps of the MCMC algorithm can be performed in parallel: simulation of $\beta$ at the $Y$ and $N$ locations (conditional on $\beta_{\mathcal{S}}$); the update of $N$ in each of the $L$ squares. This means that parallelization leads to huge computational gains and the running time of the algorithm is heavily influenced by the number of cores available.

Considering all the algorithms and issues described in this section, the MCMC algorithm to sample from the posterior distribution for the level-set Cox process model is as presented in Appendix B. The parallelizable steps are indicated as such.

\subsection{Spatiotemporal extension}

The level-set Cox process model proposed in Section \ref{sec_model} can be extended to a spatiotemporal context in which the data can be seen as a time series of point processes in a common space $S$ in discrete time. The temporal dependence is defined by a spatiotemporal Gaussian process and, possibly, a temporal structure for the level parameters of the IF. Conditional on those components, the observed Poisson process is independent among different times. We consider a particular case of the well-known Dynamic Gaussian processes (DGP) to model the temporal dependency among the random partitions. DGPs are a wide and flexible family of spatiotemporal Gaussian processes \citep[see][]{dani4}.

Suppose that the point process $Y$ is observed at $T+1$ times - $0,\ldots,T$, where $Y_t=\{Y_t(s): s \in S\}$ is a Poisson process with IF $\lambda_{t,S}=\{\lambda_t(s): s \in S\}$.
For each time $t$, we define a finite partition $\mathbf{S}_{t,K}=\{S_{t1},\ldots,S_{t,K}\}$, $K\in \mathbb{N}$, of $S$ and a sequence $c=(c_1,\ldots, c_{K-1}) \in \mathbb{R}^{K-1}$, with $-\infty=c_0 < c_1 < \ldots < c_{K-1}< c_K=\infty$. The spatiotemporal model is defined as follows.
\begin{eqnarray}
\ds (Y_t|\lambda_{t,S}) &\stackrel{ind.}{\sim}& PP(\lambda_{t,S}),\;t=0,\ldots,T,   \\
\lambda_t(s) &=& \displaystyle\displaystyle\sum_{k=1}^{K} \lambda_{t,k} I_{t,k}(s),\;s\in S,\;t=0,\ldots,T,  \\
S_{t,k} &=& \{ s\ \in\ S;\; c_{k-1} < \beta_t(s) < c_{k}    \},\; k=1,\ldots,K,\;t=0,\ldots,T, \\
\beta=(\beta_0,\ldots,\beta_T) &\sim& DNNGP (\mu,\Sigma(\sigma^2,\tau^2),\Sigma(\xi^2,\varrho^2)),\\
c &\sim& \mathds{1}(c_1<\ldots<c_{K-1})\\
\lambda_{k}=(\lambda_{0,k},\ldots\lambda_{T,k}) &{\stackrel{\textrm{ind}}{\sim}}& NGAR1(a_{0,k}, b_{0,k}, w_k, a_k), \; k=1,\ldots,K,
\end{eqnarray}
where $I_{t,k}(s)$ is the indicator of $\{s \in S_{t,k}\}$, $DNNGP$ is a dynamic NNGP and $NGAR1$ is an order 1 non-Gaussian non-linear autoregressive model. We consider a DNNGP of the form:
\begin{eqnarray}
  \beta_0 &\sim& NNGP(\mu,\tilde{\Sigma}(\sigma^2,\tau^2)), \label{stm1}\\
  \beta_t(s) &=& \beta_{t-1}(s) + \zeta_t(s),\;s\in \mathcal{S},\;t=1,\ldots,T, \label{stm2}\\
  \zeta_t &\sim& NNGP(0,\tilde{\Sigma}(\xi^2,\varrho^2)),\label{stm3} \\
  (\beta_t(s)|\beta_{t,\mathcal{S}}) &\sim& NNGP(\mu,\tilde{\Sigma}(\sigma^2,\tau^2)),\;s\in S\setminus\mathcal{S},\;t=1,\ldots,T, \label{stm4}
\end{eqnarray}
As in the case of the spatial model, we set $\mu=0$, $\sigma^2=1$ and fix $\tau^2$ at a suitable value. Parameters $\xi^2$ and $\varrho^2$ are also fixed such that $\xi^2\leq\sigma^2$ and $\varrho^2\geq\tau^2$. The DNNGP in (\ref{stm1})-(\ref{stm4}) differs from the one in \citet{NNGP} in the distribution of $(\beta_t(s)|\beta_{t,\mathcal{S}})$. \citet{NNGP} consider the formulation in (\ref{stm1})-(\ref{stm3}) for all $s\in S$. Our modification is motivated by the fact that we require to sample $\beta$ at different sets of locations ($N$ and $Y$) for each time $t$ and the DNNGP from \citet{NNGP} would be computationally inefficient in this case. Note that, under our DNNGP, the temporal dependence is explicit only in $\mathcal{S}$ and, conditional on $\beta$ at those locations, the remainder in conditionally independent w.r.t. time.

We define the following $NGAR1$ model.
\begin{eqnarray}
  \lambda_{0} &\sim& RG(a_{0}, b_{0}, \rho, \nu), \label{stm5}\\
  \lambda_{t,k} &=& w_{k}^{-1}\lambda_{t-1,k}\epsilon_{t,k},\;t=1,\ldots,T,\;k=1\ldots,K,\label{stm6} \\
  \epsilon_{t,k} &\sim& Beta(w_ka_k,(1-w_k)a_k),\label{stm7}
\end{eqnarray}
where $RG$ is the repulsive gamma prior defined in (\ref{prior_lambda}) and parameters $a_{0}$, $b_{0}$, $\rho$, $\nu$, $w_k$ and $a_k$ are fixed at suitable values. This $NGAR1$ model imposes a random walk type structure to the logarithm of $\lambda_{t,k}$ and is inspired by the non-Gaussian state space model proposed in \citet{smithmilner}.

A temporal structure can be considered to model both the random partitions and the levels of the IF as in the model above or to model only the former, in which case the $\lambda_{t,k}$ parameters are all independent with repulsive gamma priors.


Inference for the spatiotemporal level-sex Cox process model requires some adaptations to the MCMC algorithm proposed for the spatial model. First, we need to define one pseudo-marginal estimator for the likelihood of $Y$ in each time $t$ and we shall define the respective auxiliary variables as $N^*=(N_{0}^*,\ldots,N_{T}^*)$. Now, each $N_{t}^*$ is an independent unit rate Poisson process on the infinite height cylinder with base $S$ and $N_t=g(N_{t}^{*},\lambda_{t}^*)$, $\lambda_{t}^*= (\delta_t\lambda_{t,M} -\lambda_{t,m})$, $\lambda_{t,M}=\underset{k}{\max}\{\lambda_{t,k}\}$ and $\lambda_{t,m}=\underset{k}{\min}\{\lambda_{t,k}\}$. Furthermore, whenever $\beta$ is to be sampled retrospectively on the update steps of the $\lambda_k$'s and $N^*$, it is sampled from the spatiotemporal NNGP prior and, therefore, conditional on all the locations of $\beta$ already sampled at all times 0 to $T$, considering the respective conditional independence structure.

The sampling step of $N^*$ is performed analogously to the spatial case at each time $t$, independently. Parameter $c$ uses the same proposal of the spatial case and accepts a move
$c\rightarrow\ddot{c}$ with probability
\begin{equation}
\ds \alpha_{c}=1 \wedge
\left(
\displaystyle\prod_{t=0}^{T}\displaystyle\prod_{k=1}^{K} \left( \frac{\delta_t\lambda_{tM} - \lambda_{t,k}}{\delta_t\lambda_{tM} - \lambda_{tm}}\right)^{|\ddot{N}_{t,k}|-|N_{t,k}|} \left(\lambda_{t,k} \right)^{|\ddot{Y}_{t,k}|-|Y_{t,k}|}
\right).
\end{equation}

Process $\beta$ is proposed from the following spatiotemporal pCN proposal.
\begin{eqnarray}\label{mod_rw_b2}
\ddot{\beta}_{t} &=& \sqrt{1-\varsigma^2}\beta_{t}+\varsigma\varepsilon_{t},\;t=0,\ldots,T, \nonumber \\
(\varepsilon_0,\ldots,\varepsilon_T) &\sim& DNNGP (0,\Sigma(\sigma^2,\tau^2),\Sigma(\xi^2,\varrho^2)), \nonumber
\end{eqnarray}
The proposal variance $\varsigma^2$ is chosen so to have an acceptance rate of approximately 0.234 \citep{cotter}. The acceptance probability of a move $\beta\rightarrow\ddot{\beta}$ is given by
\begin{equation}
\ds \alpha_{\beta}=1 \wedge
\left(\ds\prod_{t=0}^{T}\displaystyle\prod_{k=1}^{K} \left( \frac{\delta_t\lambda_{tM} - \lambda_{t,k}}{\delta_t\lambda_{tM} - \lambda_{tm}}\right)^{|\ddot{N}_{t,k}|-|N_{t,k}|} \left(\lambda_{t,k} \right)^{|\ddot{Y}_{t,k}|-|Y_{t,k}|}\right).
\end{equation}

Finally, if no temporal dependence structure is considered for the $\lambda_k$'s, the same algorithm from the spatial model is considered independently at each time $t$ to sample $(\lambda_{t,1},\ldots,\lambda_{t,K})$. If however the $NGAR1$ prior is adopted, the $\lambda_k$'s are jointly sampled by proposing from a properly tuned Gaussian random walk.
The respective acceptance probability is given by
\begin{equation}
\ds \alpha_{\lambda}=1 \wedge \left(\left[\displaystyle\prod_{t=0}^{T}e^{-\mu(S)(\ddot{\lambda}_{t,m}-\lambda_{t,m})} \displaystyle\prod_{k=1}^{K}
\frac{\left( \frac{\delta_t\ddot{\lambda}_{t,M} - \ddot{\lambda}_{t,k}}{\delta_t\ddot{\lambda}_{t,M} - \ddot{\lambda}_{t,m}}\right)^{|\ddot{N}_{t,k}|}}
{\left( \frac{\delta_t\lambda_{t,M} -\lambda_{t,k}}{\delta_t\lambda_{t,M} - \lambda_{t,m}}\right)^{|N_{t,k}|}} \left(\frac{\ddot{\lambda}_{t,k}}{\lambda_{t,k}}\right)^{|Y_{t,k}|}\right]
\frac{\pi(\ddot{\lambda})}{\pi(\lambda)}\right),
\end{equation}
where $\pi(\lambda)$ is the density of the $NGAR1$ prior.

\subsection{Prediction}

It is often the case that the analysis of point process phenomena also aims at predicting unknown quantities conditional on the data. In the spatial context, this may include some functional of the IF in a given region of $S$ or future replications of the observed process. In the spatiotemporal context, prediction about future times may be considered. In both cases, it is straightforward to obtain a sample from the desired posterior predictive distribution based on the output of the MCMC.

The algorithm to sample from the predictive distribution of a function $h(\lambda_{S})$ under the spatial model consists of computing $h(\lambda_{S})$ for each sampled value of $\lambda_{S}$ in the MCMC chain (after a reasonable burn-in). If $h(\lambda_{S})$ is intractable it may still be possible to obtain a sample from the predictive distribution of an unbiased estimator of $h(\lambda_{S})$. For example, define $h(\lambda_{S})=\int_{S_0}\lambda(s)ds:=\Lambda_{S_0}$, for some known $S_0\subset S$, and $U\stackrel{ind.}{\sim}Unif(S_0)$. Then, an unbiased estimator of $h(\lambda_{S})$ is given by \citep[see][Section 4.3]{G&G}
\begin{equation}\label{eq_uestint}
\hat{h}=\mu(S_0)\lambda(U).
\end{equation}
A sample from the predictive distribution of (\ref{eq_uestint}) is obtained by sampling $U_i\sim Unif(S)$ and $\lambda(U_i)$ on each iteration of the MCMC.

To perform prediction for replications of $Y$ it is enough to simulate $Y$ conditional on each sampled value of $\lambda_{S}$ in the MCMC using a Poisson thinning algorithm.
This consist of simulating a $PP(\lambda_M)$ on $S$ and keeping each point $s$ with probability $\lambda(s)/\lambda_M$, where the value of $\lambda(s)$ is obtained by sampling $\beta(s)$ retrospectively from the GP prior on each MCMC iteration.

Now consider the full Bayesian model of a level-set Cox process $Y$ in $S$ for times $0,\ldots,T,\ldots,T+d$, $d\in\mathds{N}$, and let $y$ be a realization of the process at times $0,\ldots,T$. Define $h(Y,\lambda_d)$ to be some measurable function, in the probability space of the full Bayesian model, that depends on $(Y_t,\lambda_{t,S})$ only for times $t\in\{T+1,\ldots,T+d\}$. Then, prediction about $h(Y,\lambda_d)$ is made through the predictive distribution of $(h(Y,\lambda_d)|y)$. This is sampled by simulating $h(Y,\lambda_d)$, conditional on the output of the MCMC on each iteration, based on the following identity.\\
\begin{equation}\label{eq_pred}
\pi(h(Y,\lambda_d)|y)=\int \pi(h(Y,\lambda_d)|\lambda_{0:T},y)\pi(\lambda_{0:T}|y)d\lambda_{0:T}.
\end{equation}
Appealing examples of $h(Y,\lambda_d)$ include:
\begin{enumerate}[i.]
  \item $\ds(\lambda_{T+1,S},\ldots,\lambda_{T+d,S})$;
  \item $\ds \Lambda_{S,d}=\int_{S}\lambda_t(s)ds$, for $t=T+1,\ldots,T+d$;
  \item $\ds(Y_{T+1},\ldots,Y_{T+d})$.
\end{enumerate}

\section{Simulated examples}\label{sec_sim}

We perform a series of simulation studies to investigate the main issues regarding the methodology proposed in this paper. First, we present a sensitive analysis w.r.t. the prior specifications of the covariance function of the Gaussian process prior for two models that differ in terms of the size of the dataset. Then, we explore the choice of the number of levels $K$. All the examples presented in this paper are implemented in Ox \citep{doornik2009object} and run in an i7 3.50GHz processor with 6 cores (12 threads) and 16GB RAM. Codes for the spatial and spatiotemporal models are available at \emph{https://github.com/fbambirra/MCMC-LSCP}.

In all the simulations, we consider the initial values $c=(-0.5,0.5)$. The initial values of $\beta$, $\lambda$ and $N$ are set as described in Section \ref{susec_issues}. The repulsive Gamma prior is adopted for $\lambda$ with $\alpha_k=1.2$, $\eta_k=0.04$, for all $k$, $\rho=1$ and $\nu=3$. The value of $\delta$ is set to have $|N|\approx6000$. The efficiency of the proposed methodology is investigated in terms of estimation and computational cost.

\subsection{Sensitivity analysis}\label{subsecsens}

We consider three examples with $K=3$ and the same partition for the piecewise structure of the IF. The three examples differ in the values of $\lambda$. We call them examples 1, 2 and 3, having true $\lambda$ equals to $(1,\;4,\;12)$, $(3,\;20,\;50)$ and $(20,\;50,\;100)$, respectively. We consider one replication of example 3, which has 5267 observations, to illustrate the applicability of our methodology to large datasets and 10 replications of examples 1 (around 500 observations) and 2 (around 2200 observations). Figure \ref{figsap5} presents the true IF for examples 1 and 2 and Figure \ref{figsap4b} in Appendix D presents the true IF for example 3.

In the first sensitivity analysis, we compare the results for one replication of examples 1 and 2, for $\tau^2=0.5$, 1, and 2. Figure \ref{figap1} in Appendix C shows a plot of the three respective correlation functions. Comparison is performed in terms of the computational cost and estimation of the intensity function. Tables \ref{tab1} and \ref{tab2} show some of the results and Figure \ref{figsap5} shows the estimated intensity function. We can clearly see that as the value of $\tau^2$ increases, the estimated IF gets smoother, as expected. Values 1 and 2 provided quite a good recovery of the true IF, with the latter performing a bit better in example 1 and the former performing a bit better for example 2 (see the left-bottom part of the area with the highest IF). To compare the results for the replications of both examples, we choose a common value $\tau^2=1$ to analyze the 10 replications of examples 1 and 2 and the one replication of example 3. The estimated IFs for all the replications are presented in Figures \ref{figsap7} and \ref{figsap8} in Appendix D. Some posterior statistics are presented in Table \ref{tab3}.

Each MCMC chain runs for 300 thousand iterations and the average running time is around 15.5 hours for example 1 and around 19.5 hours for example 2, with very small variations among replications and different values of $\tau^2$, and around 48 hours for example 3. The effective sample sizes (ESS) reported here are computed with the R package CODA \citep{CODA}. Trace plots and autocorrelation plots for the case where $\tau^2=1$ are presented in Figure \ref{figap6} in Appendix D and strongly suggest the convergence of the algorithm.

\begin{figure}[h!]
	\centering{
		\includegraphics[width=0.9\linewidth]{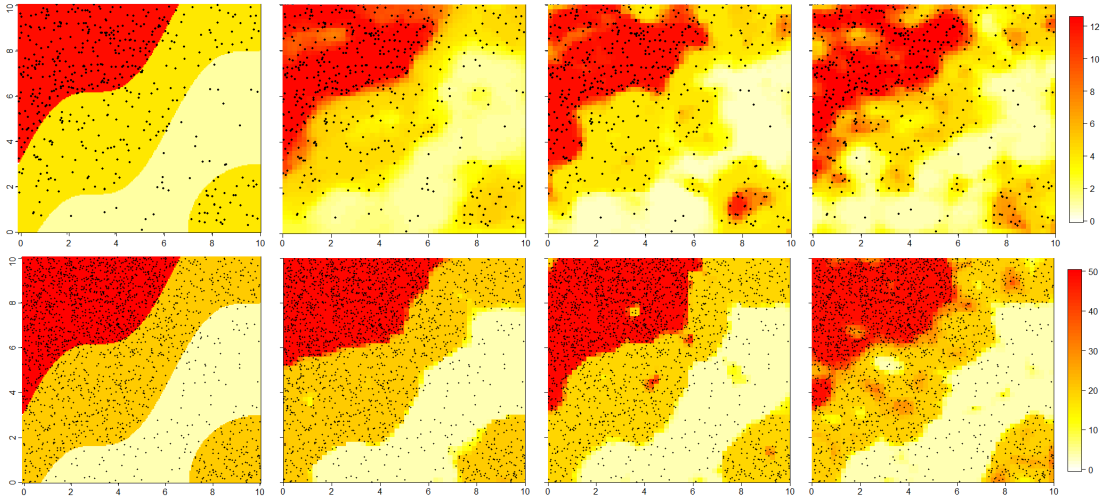}
		\caption{True IF (1st column) and its posterior mean for examples 1 (top) and 2 (bottom), for $\tau^2=$ 0.5 (2nd column), 1 (3rd column) and 2 (4th column).}
		\label{figsap5}}
\end{figure}

\begin{table}[h!]
  \centering \setlength\tabcolsep{2pt}
\begin{tabular}{|c|c|c|c|}
  \hline
  $\tau^2$ & 2 & 1 & 0.5 \\ \hline
  ESS            &  765  &  693  &  559   \\ \hline
  $\lambda_1=1$  &  1.15(0.29) &  0.67(0.18) &  0.91(0.18)  \\
  $\lambda_2=4$  &  4.51(0.46) &  3.99(0.31) &  4.50(0.62)  \\
  $\lambda_3=12$ & 12.18(0.80) & 11.97(0.74) & 12.38(0.79)  \\
  \hline
\end{tabular}
  \caption{Results for example 1. Second row reports the ESS of the log pseudo-marginal likelihood. The remaining rows show the posterior mean and standard deviation of the $\lambda_k$ parameters.}\label{tab1}
\end{table}

\begin{table}[h!]
  \centering \setlength\tabcolsep{2pt}
\begin{tabular}{|c|c|c|c|}
  \hline
  $\tau^2$ & 2 & 1 & 0.5 \\ \hline
  ESS            &  906 & 1386 &   1876  \\ \hline
  $\lambda_1=3$  &  3.52(0.34) &  3.60(0.34) &  5.82(0.40)  \\
  $\lambda_2=20$ & 20.04(0.72) & 19.09(0.72) &  21.44(0.67)  \\
  $\lambda_3=50$ & 49.19(1.49) & 48.45(1.44) &  48.91(1.50)  \\
  \hline
\end{tabular}
  \caption{Results for example 2. Second row reports the ESS of the log pseudo-marginal likelihood. The remaining rows show the posterior mean and standard deviation of the $\lambda_k$ parameters.}\label{tab2}
\end{table}

\begin{table}[h!]
  \centering \setlength\tabcolsep{2pt}
\begin{tabular}{|c|c|c|c|c|c|c|}
  \hline
             &   \multicolumn{2}{c|}{Example 1}  & \multicolumn{2}{c|}{Example 2}  & \multicolumn{2}{c|}{Example 3}  \\ \hline
   aver. ESS &   \multicolumn{2}{c|}{596(237)}  & \multicolumn{2}{c|}{1103(498)}   & \multicolumn{2}{c|}{889} \\ \hline
              & True & Est. & True & Est & True & Est\\ \hline
  $\lambda_1$ & 1  & 0.84(0.21) / 0.22(0.09)   & 3  &  3.47(0.64) / 0.33(0.04) &  20 & 21.88 / 0.88 \\
  $\lambda_2$ & 4  & 4.13(0.70) / 0.42(0.15)   & 20 & 20.27(0.81) / 0.73(0.08) &  50 & 48.66 / 1.11 \\
  $\lambda_3$ & 12 & 12.59(1.32) / 0.83/(0.11) & 50 & 50.70(1.56) / 1.42(0.14) & 100 & 100.34 / 1.55 \\
  \hline
\end{tabular}
  \caption{Results for the 10 replications of examples 1 and 2 and for the 1 replication of example 3. Second row reports the mean and s.d. of the ESS of the log pseudo-marginal likelihood over the 10 replications. The remaining rows show the mean and standard deviation, over the 10 replications, of the posterior mean and s.d. of the $\lambda_k$ parameters.}\label{tab3}
\end{table}

\subsection{Model fit}\label{model_fit}

We now explore the issue of choosing the number of levels $K$. We fit the LSCP model for one replication of example 1 and one of example 2 for two values of $K$, with $\tau^2=1$. We compare the results with those obtained by the methodology proposed in \cite{geng21}, who use a mixture of finite mixtures model to detect the number of clusters $K$ and estimate the IF in each of them after discretizing the space. We consider 3 levels of discretization - $10\times10$, $15\times15$ and $20\times20$.

For the dataset of example 1, their algorithm estimates $K=3$ for $10\times10$ and $K=2$ for $15\times15$ and $20\times20$. For example 2, it estimates $K=4$ for $10\times10$ and $K=3$ for $15\times15$ and $20\times20$. The IF estimates are shown in Figure \ref{figsap10} in Appendix D.

We fit the LSCP model to the dataset of example 1 for $K=2$ and 3 and to the dataset of example 2 for $K=3$ and 4. In our case, the models are compared via DIC \citep{DIC2002}. The values of the DIC are -952.11 for $K=2$ and -1011.844 for $K=3$, for example 1, and -10319.6 for $K=3$ and -10278.7 for $K=4$, for example 2, which correctly indicates the respective true models. The estimates of the IF are shown in Figure \ref{fig10}.

\begin{figure}[h!]
	\centering{
		\includegraphics[width=0.95\linewidth]{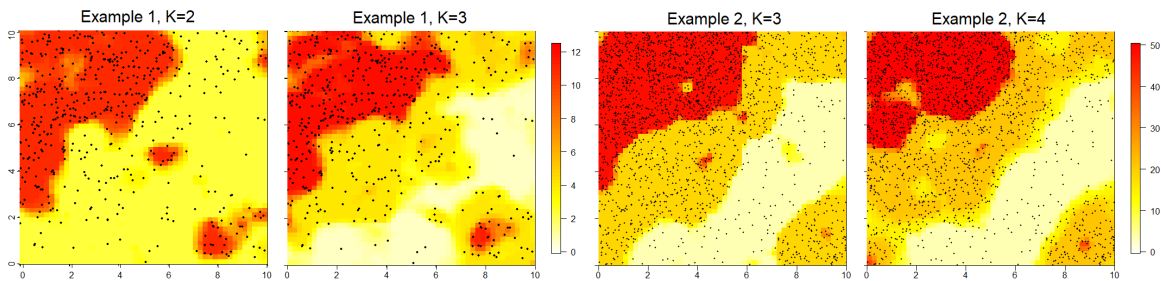}
		\caption{Estimated IF for example 1 with $K=2$ and $K=3$ and for example 2 with $K=3$ and $K=4$.}
		\label{fig10}}
\end{figure}

\begin{table}[h!]
  \centering
\begin{tabular}{|c|c|c|c|c|}
  \hline
             & \multicolumn{2}{c|}{Example 1} & \multicolumn{2}{c|}{Example 2} \\ \hline
             & $K=2$  & $K=3$  & $K=3$  & $K=4$ \\ \hline
 $\lambda_1$ &  2.17(0.20) &  0.67(0.18) &  3.60(0.34) & 3.37(0.39) \\
 $\lambda_2$ & 10.84(0.65) &  3.99(0.31) & 19.09(0.72) & 13.76(1.18) \\
 $\lambda_3$ &             & 11.97(0.74) & 48.45(1.44) & 21.45(0.84) \\
 $\lambda_4$ &             &             &             & 50.05(1.45) \\
  \hline
\end{tabular}
  \caption{Results for the sensitivity analysis regarding the specification of $K$. Posterior mean and standard deviation of the $\lambda_k$ parameters.}\label{tab3}
\end{table}

\subsection{Comparison to discrete approximation method}\label{disc_approx}

To illustrate the advantages of the exact approach of the methodology proposed in this paper, we compare it to a discretized version of the LSCP model for different levels of discretization - $20\times20$, $50\times50$ and $100\times100$. We consider the same discrete approximation as in \citet{hildeman2017level} but with a (discretized) NNGP prior for $\beta$ and the respective pCN proposal to update this coordinate via MH, and the repulsive prior for the $\lambda_k$'s. We compare the results for one of the replications of example 2 and one of the applications presented in Section \ref{sec_app_spt}. Results are shown in Figures \ref{figsap9a} and \ref{figsap9b}.

\begin{figure}[h!]
	\centering
		\includegraphics[width=0.9\linewidth]{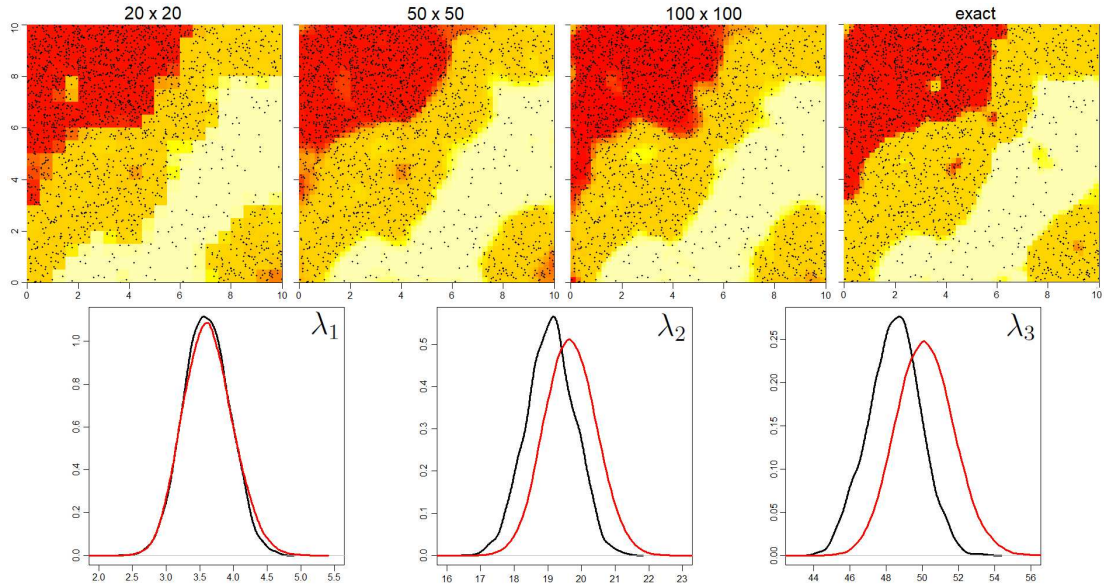}
		\caption{Discrete approximation results for example 2. Top: estimated IF for lattices $20\times20$, $50\times50$ and $100\times100$, and with the exact method. Bottom: empirical posterior density of $\lambda$ for $100\times100$ (red) and for exact method (black).}
		\label{figsap9a}
\end{figure}

\begin{figure}[h!]
	\centering
		\includegraphics[width=0.9\linewidth]{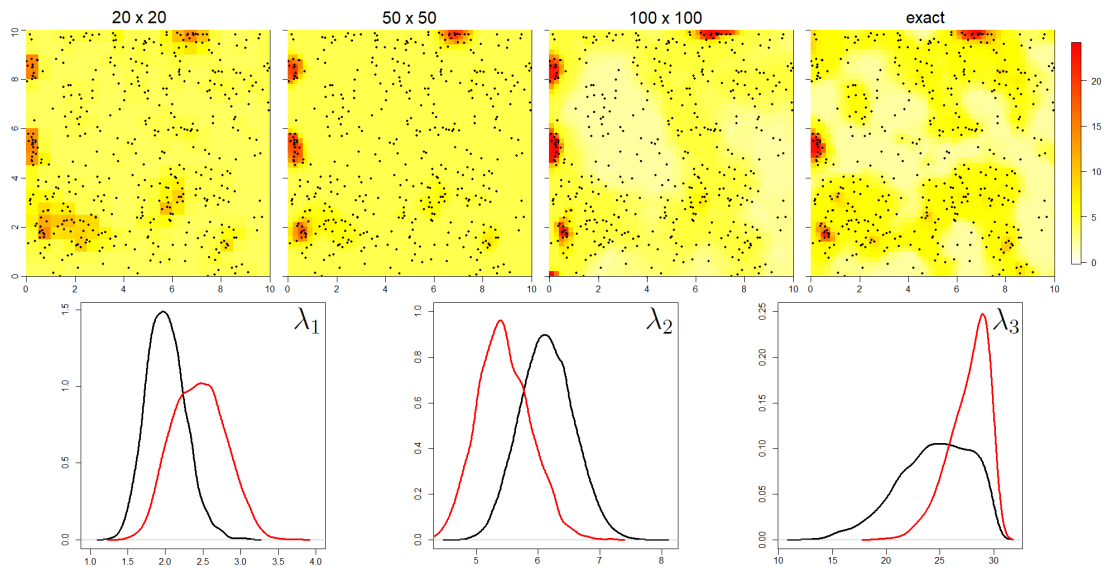}
		\caption{Discrete approximation results for the white oak example. Top: estimated IF for lattices $20\times20$, $50\times50$ and $100\times100$, and with the exact method. Bottom: empirical posterior density of $\lambda$ for $100\times100$ (red) and for exact method (black).}
		\label{figsap9b}
\end{figure}

The MCMC for the discretized method has to run for around 900 thousand iterations to obtain a reasonable MC sample for inference. This means a running time of around 25 hours for both examples with a lattice of $100\times100$. Results show that there is still a significant difference between the posterior distributions of the approximate and exact methods.

\section{Applications}\label{sec_app}

We apply the LSCP model to analyze some real point process datasets - three spatial and one spatiotemporal. The three spatial datasets consist of: 1. locations of white oak trees in some region in the USA; 2. locations of particles in a bronze filter; 3. locations of fires in a region of New Brunswick, Canada, for a period of 12 years. The spatiotemporal dataset also considers the locations of fires in New Brunswick, but disaggregates the data per periods of 3 years. All the datasets are available in the R package \texttt{spatstat} \citep{baddeley2015spatial}.

We consider the initial values $c=0$, $c=(-0.5,0.5)$ and $c=(-0.7,0,0.7)$, for $K=2$, 3 and 4, respectively. The initial values of $\beta$, $\lambda$ and $N$ are set as described in Section \ref{susec_issues}. The repulsive Gamma prior is adopted for $\lambda$ with $\alpha_k=1.2$, $\eta_k=0.04$, for all $k$, and $\nu=3$. Parameters $\rho$ from the RG prior and the range parameter $\tau^2$ vary among the examples as follows: white oak - $\rho=5$ and $\tau^2=0.5$; bronze filter - $\rho=5$ and $\tau^2=1$; fires - $\rho=1$ and $\tau^2=0.5$. Those values are based on the empirical analysis of the kernel smoothing estimates of the IF (see Figure \ref{fig1}) in terms of the levels and smoothness of the IF expected to provided a good fit. The value of $\delta$ is set to have $|N|\approx6000$.

For the first example, we also present an analysis comparing the level-set Cox process model to a Cox process model in which the IF is a continuous function of a latent Gaussian process. The latter is proposed in \citet{G&G}, who also present an exact methodology to perform Bayesian inference. Their model assumes $\lambda(s)=\lambda^*\Phi(\beta(s))$, where $\lambda^*$ is an unknown parameter, $\beta$ is a Gaussian process and $\Phi$ is the standard normal c.d.f.

\subsection{Spatial examples}\label{sec_app_spt}

We analyze a dataset regarding the locations of white oak trees in Lansing Woods, Michigan. The data consists of the location of 448 white oaks in an area of $924\times924$ feet, which we rescale to $(0,10)\times(0,10)$.

An empirical analysis based on the kernel smoothing estimation of the IF suggested that $K=3$ would be a suitable choice. Indeed, a model with $K=4$ was also fit but the area of one of the four regions converged to zero along the MCMC chain. The largest value of $\lambda$ is truncated a priori to be smaller than 30 when $K=3$. Figure \ref{fig6} shows the posterior mean and mode of the IF. The latter defines the partition using its pointwise mode and colors each region with the posterior mean of the respective $\lambda_k$.

\begin{figure}[h!]
	\centering{
		\includegraphics[width=0.9\linewidth]{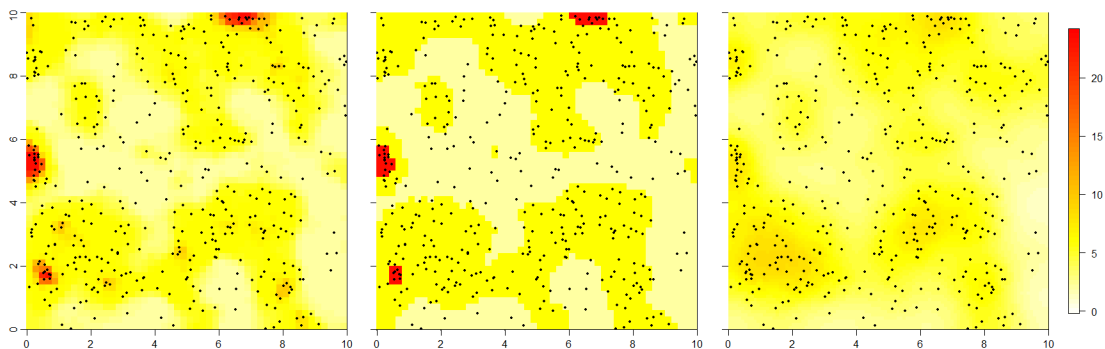}
		\caption{Posterior mean (left) and mode (middle) of the IF under the LSCP model and posterior mean (right) of the IF under the continuous IF model for the white oak example.}
		\label{fig6}}
\end{figure}

We also consider the prediction of the integrated IF in the whole observed domain and in two regions - $S_1=(5\;,\;7)\times(8\;,\;10)$ and $S_2=(8\;,\;10)\times(4.5\;,\;6.5)$, see  Table \ref{tab4}.

The results for the two models are considerably different in some aspects of the estimated IF. Generally speaking, and as expected, the estimate is smoother for the continuous IF model. The repulsive prior for the IF values in the LSCP pushes those values apart and estimates a small cluster (with a mean area around 1.66) with a much higher IF. All the 3 predicted functions of the IF have a smaller predictive variance for the LSCP, but also a slightly larger bias for the point estimates (posterior mean). If we combine the bias and variance through the expected quadratic error - $E_{\theta}[(h(\lambda_S)-true)^2]$, this is smaller for the LSCP for $\Lambda_{S}$ (70\%) and $\Lambda_{S_1}$ (88\%) and smaller for the continuous IF model for $\Lambda_{S_2}$ (92\%).

\begin{table}[h!]
	\centering
	\begin{tabular}{|c|c|c|c|}
		\hline
		&  & LSCP & Continuous IF model\\
		\hline		
		$\Lambda_{S}$   & 448 & 447.22 / 20.50 (414,482) - 421.20 & 448.44 / 24.14 (408,490) - 594.24 \\
		$\Lambda_{S_1}$ & 27  & 29.18 / 3.62 (22.91,35.18) - 17.92 & 25.35 / 4.18 (18.65,32.45) - 20.19 \\
		$\Lambda_{S_2}$ & 9   & 11.47 / 1.86 (8.65,14.87) - 9.62 & 9.98 / 2.82 (5.64,14.84) - 8.88 \\ \hline
	\end{tabular}
	\caption{Statistics of the posterior predictive distribution of the estimator in (\ref{eq_uestint}). Each cell shows: Mean / s.d. (95\% C.I.) - expected quadratic error.}
	\label{tab4}
\end{table}

The second example considers the locations of 678 particles observed in a longitudinal plane section of $18\times7$ mm through a gradient sinter filter made from bronze powder. The original area is rescaled to $(0,10)\times(0,4)$ and an empirical analysis via kernel smoothing suggests $K\approx4$. We fit the model for $K=3$ and 4 but area of one of the four regions converges to zero along the MCMC when $K=4$. The estimated IF for $K=3$ is shown in Figure \ref{fig7b}, which also brings the extra information about the radius of each particle. Although this information is not used in the analysis, the clear relation between radius and particle concentration was captured by the IF estimate.

Finally, the third example considers the locations of 2313 fires in a rectangular region (rotated 90$^o$ to the left and rescaled to $8.5\times$10) containing most of the area of New Brunswick, Canada, from 1992 to 2003. We fit the LSCP for $K=3$ and $K=4$ but the area of one of the 4 regions converges to zero along the MCMC when $K=4$. The estimated IF is presented in Figure \ref{fig7b}.

\begin{figure}[h!]
	\centering{
		\includegraphics[width=0.9\linewidth]{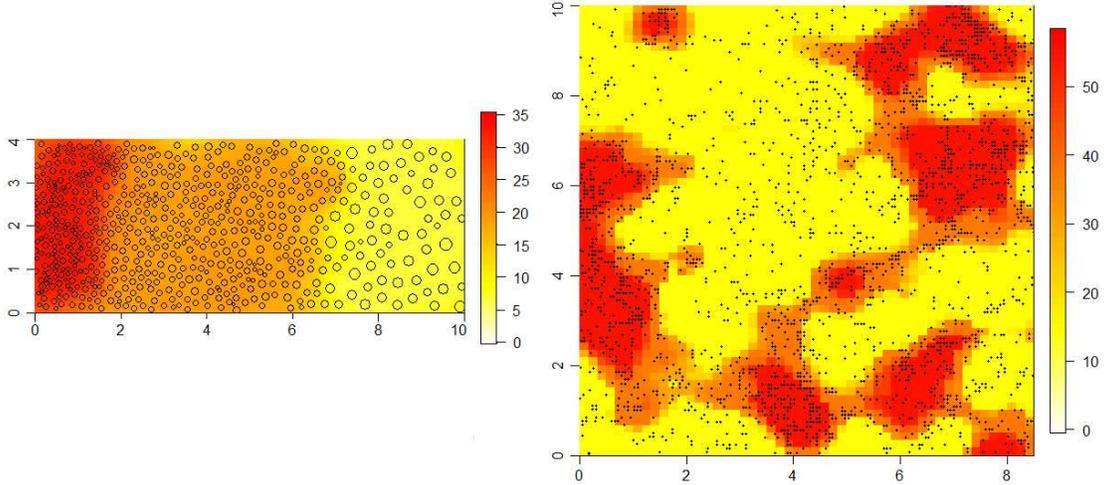}
		\caption{Posterior mean of the IF for the bronze filter example (left) and for the New Brunswick fires example (right).}
		\label{fig7b}}
\end{figure}

The estimated levels of the IF for all the 3 examples are presented in Table \ref{spex_tab}.

\begin{table}[h!]
	\centering
	\begin{tabular}{|c|c|c|c|}
		\hline
		  & White oak & Bronze filter & NB fires \\
		\hline		
		$\lambda_{1}$ & 22.48(4.63) & 33.27(2.86) & 55.11(1.95) \\
		$\lambda_{2}$ &  6.07(0.42) & 18.62(1.15) & 37.45(1.72) \\
		$\lambda_{3}$ &  1.97(0.25) &  6.47(0.79) & 13.40(0.53) \\ \hline
	\end{tabular}
	\caption{Posterior mean and standard deviation of $\lambda$ for the three spatial application.}
	\label{spex_tab}
\end{table}

The results for $K=4$ in all three examples show the impact of the repulsive gamma prior used for the $\lambda_k$ parameters. It penalizes scenarios with similar values of $\lambda_k$'s and, as the values are pushed apart, estimates one of the areas to be zero.

\subsection{Spatiotemporal example}

We consider the New Brunswick fires dataset for the years 1992 to 2003 aggregating every 3 years as one time $t$ in the model. The number of fires per each interval of 3 years is 414, 385, 450 and 415, respectively. We fit the model in (\ref{stm1})-(\ref{stm4}) with $K=3$ and perform prediction of the IF and its integral for the interval of 3 years 2004-2006. We set $\rho=5$, $\tau^2=0.5$, $\xi^2=1$, $\varrho^2=0.5$. Independent repulsive gamma priors are assumed for each $\lambda_{t,k}$ and the NGAR1 prior, with $w_k=0.5$ and $(a_1,a_2,a_3)=(5,15,30)$, for all $k$, is assumed between the respective levels from times 3 and 4 in order to perform prediction for the latter. All the other specifications are as chosen for the spatial examples. Results are shown in Table \ref{spex_tab2} and Figures \ref{fig8} and \ref{fig9}.

\begin{table}[h!]
	\centering
	\begin{tabular}{|c|c|c|c|c|c|}
		\hline
		  & $T=0$ & $T=1$ & $T=2$ & $T=3$ & $T=4$ (prediction) \\
		\hline		
		$\lambda_{1}$ & 19.17(2.40) & 13.95(1.66) & 16.92(1.80) & 14.24(1.11) & 12.73(6.83) \\
		$\lambda_{2}$ &  7.78(0.48) &  6.18(0.41) &  6.52(0.53) &  5.31(0.47) & 3.96(1.77) \\
		$\lambda_{3}$ &  1.70(0.20) &  2.51(0.22) &  2.84(0.25) &  2.02(0.21) & 1.42(0.51) \\ \hline
	\end{tabular}
	\caption{Posterior mean and standard deviation of $\lambda$ for the spatiotemporal application.}
	\label{spex_tab2}
\end{table}

\begin{figure}[h!]
	\centering{
		\includegraphics[width=0.9\linewidth]{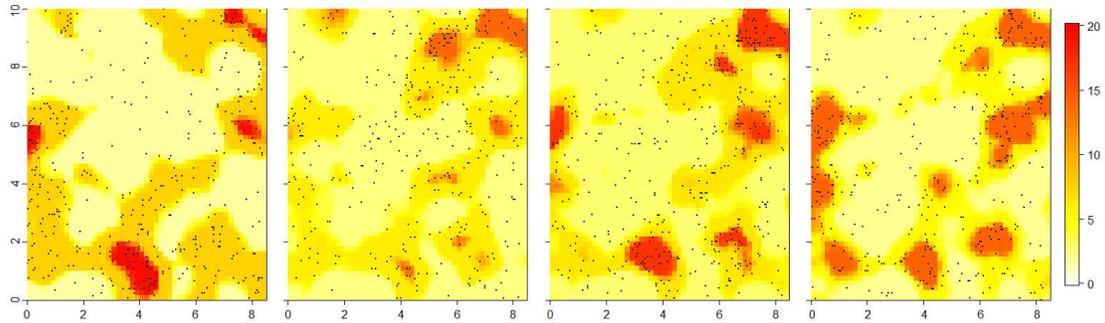}
		\caption{Posterior mean of the IF, at time 0 to 3, for the spatiotemporal example.}
		\label{fig8}}
\end{figure}

\begin{figure}[h!]
	\centering{
		\includegraphics[width=0.8\linewidth]{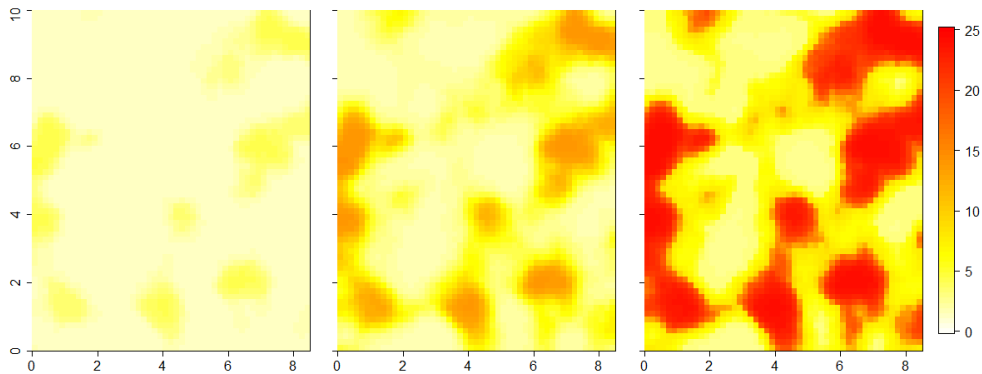}
		\caption{Predictive mean (middle) and pointwise 95\% credibility interval (left and right) of the IF in year 2004 for the spatiotemporal example.}
		\label{fig9}}
\end{figure}

\section{Conclusions}\label{sec.conc}

This paper proposed a novel methodology to perform exact Bayesian inference for a class of level-set Cox processes in which the intensity function is piecewise constants. The model is flexible enough to accommodate any smooth partition structure and aims at providing a more parsimonious alternative to Cox process models with continuously varying IF. The methodology is exact in the sense of not involving discrete finite-dimensional approximations and is the first one with this feature for the class of LSCP models.

The inference is performed via an infinite-dimensional pseudo-marginal MCMC algorithm. The MCMC chain has the exact posterior distribution of all the unknown components of the model as its invariant distribution. This means that only MCMC error is involved despite the intractability of the likelihood function and infinite dimensionality of the parameter space. Retrospective sampling and pseudo-marginal Metropolis are used to circumvent the infinite dimensionality and intractable likelihood problems, respectively. Efficient proposal distributions are carefully devised for the latent Gaussian process component and the pseudo-marginal auxiliary variable. Computational cost issues are mitigated by adopting a NNGP approach for the latent Gaussian process and by adding a virtual retrospective sampling step to the MCMC algorithm that deletes extra sampled locations of the GP component.

A variety of issues related to the efficiency of the proposed MCMC algorithm are discussed and empirically explored through simulations. Model fitting regarding the choice of the number of levels for the IF is also explored. Results show a considerably good performance of the proposed methodology.

Finally, a spatiotemporal version of the level-set Cox process model is introduced and applied to a real dataset regarding fires in the province of New Brunswick, Canada.

Some interesting directions may be pursued as future work. For example, the sensitivity of the proposed methodology to the choice of the covariance function of the latent GP. Any valid covariance function can be used within the proposed methodology, so it is natural to question if the partition estimation may benefit from more complex structures such as non-stationary ones. Three possible extensions of the proposed methodology may also be considered. First, estimating the number of levels $K$ may be particularly useful in applications in which selecting $K$ is not a trivial task. Second, the methodology from this paper can be merged with that from \citet{G&G} so that the IF is a continuous function of independent Gaussian processes, conditional on the partition, and the inference methodology is still exact. Third, one may consider the use of spatial covariates in the intensity function, for example, $\lambda(s)=\lambda_{k,0}+\lambda_{k,1}X_1(s)$, for some covariate $X_1$.


\section*{Acknowledgements}
The first author would like to thank FAPEMIG - Grant PPM-00745-18 and CNPq - Grant 310433/2020-7, for financial support. The second author would like to thank CAPES for financial support. The authors would like to thank Gareth Roberts for insightful discussions about the MCMC algorithm.

\bibliographystyle{Chicago}
\bibliography{biblio1}

\section*{Appendix A - Proofs}

\subsection*{Proof of Proposition \ref{prop1}}

	Let $I_{nk}= I_{k}(s_n)$ be the indicator of $s_n \in S_k$, where $s_n$ is the $n$-th point from $N$ and $I=(I_1,\ldots,I_{|N|})$, where  $I_n=(I_{n1},\ldots,I_{nK}) \sim Mult\left(1,\frac{\mu_1}{\mu(S)},\ldots,\frac{\mu_K}{\mu(S)}\right)$. Therefore, $E(I_{nk})=\frac{\mu_k}{\mu(S)}$ and $\mu(S) I_{nk}$ is an unbiased estimator of $\mu_k$. Then,
	\begin{eqnarray}
	E_{|N|,I}[\hat{M}]&=&E_{|N|,I}\left[e^{-\mu(S)\lambda_{m}} \displaystyle\prod_{k=1}^{K} \left( \frac{\delta\lambda_{M} - \lambda_k}{\delta\lambda_{M} - \lambda_{m}}\right)^{|N_k|}  \right] \nonumber \\
	&=& E_{|N|,I} \left[e^{-\mu(S)\lambda_{m}} \displaystyle\displaystyle\prod_{n=1}^{|N|} \left( \frac{\delta\lambda_{M} - \sum_{k=1}^{K} I_{nk}\lambda_k}{\delta\lambda_{M} - \lambda_{m}}\right)    \right] \nonumber \\
	&=& e^{-\mu(S)\lambda_{m}} E_{|N|} \left[ \left(\frac{\mu(S)\delta\lambda_{M}-\sum_{k=1}^{K} \mu_k \lambda_k}{\mu(S)(\delta\lambda_{M} - \lambda_{m})}\right)^{|N|} \right] \nonumber \\
	&=& e^{-\mu(S)(\lambda_{m}+\delta\lambda_M-\lambda_m)} \sum_{j=0}^{\infty}  \frac{\left(\mu(S)\delta\lambda_{M}-\sum_{k=1}^{K} \mu_k \lambda_k\right)^j}{j!} \\ \nonumber
	&=& e^{-\sum_{k=1}^{K} \mu_k \lambda_k} = M.
	\end{eqnarray}

\subsection*{Proof of Proposition \ref{prop2}}

We shall compute the variance of $\hat{M_1}=\displaystyle\prod_{k=1}^{K} \left( \frac{\delta\lambda_{M} - \lambda_k}{\delta\lambda_{M} - \lambda_{m}}\right)^{|N_k|}$.
\\
We use the basic probability result that if $X\sim Poisson(\lambda)$, then $E[a^{nX}]=\exp\{-\lambda(1-a^n)\}$, $n\in\mathds{N}$, and the fact that $|N_k|\sim Poisson(\mu_k\lambda^*)$ to compute
$E[\hat{M_1}^2]$ and $E[\hat{M_1}]$ and obtain
	\begin{eqnarray}
	Var\left(\hat{M}_1\right) &=& \exp\left\{  -\sum_{k=1}^{K}   \mu_k(\delta\lambda_{M} - \lambda_m) \left[ 1-\left(\frac{\delta\lambda_{M} - \lambda_k}{\delta\lambda_{M} - \lambda_m}\right)^2\right]\right\} \nonumber\\
	&-&\exp\left\{  -2\sum_{k=1}^{K}   \mu_k(\lambda_{k} - \lambda_m) \right\}, \nonumber
	\end{eqnarray}
	implying that
	\begin{equation}
	Var(\hat{M})=\exp\left\{-2\mu(S)\lambda_{m}\right\} Var\left(\hat{M}_1\right). \nonumber
	\end{equation}
	Finally,
	\begin{equation}
	\frac{\partial Var(\hat{M}_1)}{\partial \delta}= \exp(\kappa)\left[ \sum_{k=1}^{K}\mu_k\lambda_M\left( \left( \frac{\delta\lambda_M-\lambda_k}{\delta\lambda_M-\lambda_m} \right)^2 -2(\delta\lambda_M-\lambda_k) -1 \right) \right]<0,	\nonumber
	\end{equation}
	where $\kappa\in\mathds{R}$.

\newpage

\section*{Appendix B - The MCMC algorithm}

\begin{algorithm}[!h]
\caption{MCMC for the level-set Cox process model}\label{alg_1}
\begin{algorithmic}[1]
\Require{$K$, $L$, $\delta$, $\tau^2$, and initial values for $N$, $\beta$, $\lambda$, $c$.}
\Ensure{MCMC posterior sample of $\theta$.}
\State Simulate $\beta_{\mathcal{S}}$ from the pCN proposal in (\ref{mod_rw_b}).
\State (In parallel) Simulate $\beta$ at $Y$ and $N$ from the pCN proposal, conditional on the $\beta_{\mathcal{S}}$ simulated on the previous step.
\State Accept the proposal w.p. given in (\ref{a.p.beta}).
\For{$l = 1 \to L$}  (In parallel)
\State For the $l$-th square, propose $N$ from the pseudo-marginal proposal $q(N^*)$ and accept w.p. given in (\ref{DTPMN2}).
\State Perform virtual update (delete all the values of $\beta$ in $S\setminus\{Y,N\}$).
\EndFor
\State \textbf{end for}
\State Propose $\lambda$ from a properly tuned Gaussian random walk.
\State If the proposed value of $\lambda$ yields a value $\ddot{\lambda}^*$ larger than the current one $\lambda^*$, simulate $N^*$ between the current and proposal values of $\lambda^*$ - simulate the number of locations from a $Poisson((\ddot{\lambda}^*-\lambda^*)\mu(S))$ and distribute them uniformly in the section of the cylinder.
\State  (In parallel) If extra locations of $N^*$ are simulated in the previous step, simulate $\beta$ at those locations from the NNGP prior conditional on its values at $\{Y,N\}$.
\State Accept the proposed value of $\lambda$ w.p. given in (\ref{a.p.lambda}).
\State Perform virtual update (delete all the values of $\beta$ in $S\setminus\{Y,N\}$).
\State Propose $c$ from a properly tuned uniform random walk and accept w.p. given in (\ref{a.p.c}).
\State If enough iterations of the MCMC have been performed, stop; otherwise, go back to 1.
\end{algorithmic}
\end{algorithm}

The virtual update is performed if the set of the unveiled locations of $\beta$ in $S\setminus\{Y,N\}$ is not empty. The tuning of the random walk proposal is performed in a pre-run of the algorithm.

\newpage

\section*{Appendix C - Plots}\label{secsim}

\begin{figure}[!h]
\centering
   \includegraphics[width=0.9\linewidth]{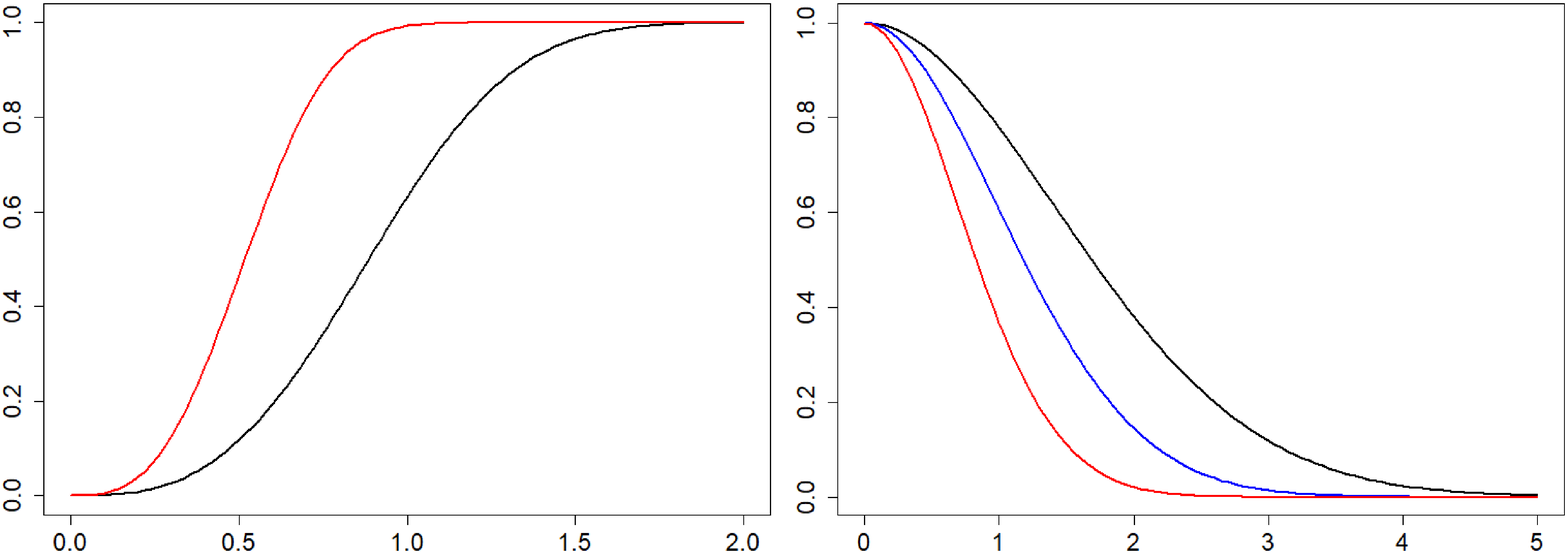}
\caption{Left: penalizing factor $r(x)=\left(1-\exp\left\{-\rho x^{\nu}\right\}\right)$ of the RG prior. Right: powered exponential covariance function with $\gamma=1.95$ and $\tau^2=0.4$ (red), 1 (blue) and 2 (black).}\label{figap1}
\end{figure}

\section*{Appendix D - Further results from the simulations}\label{secsim}

Figure \ref{figap6} shows the trace plots and ACF plots of the pseudo-marginal likelihood function for the examples in Section \ref{subsecsens}.

\begin{figure}[!h]
\centering
   \includegraphics[width=1\linewidth]{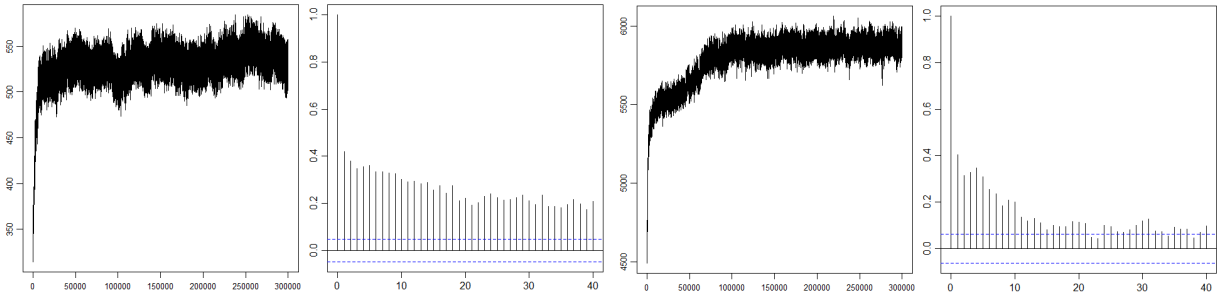}
\caption{Trace plots and ACF plots of the log pseudo-marginal likelihood function for one replication of example 1 (left) and one for example 2 (right). The ACF plots are based on a sub-sample of the chain with a lag of 100.}\label{figap6}
\end{figure}

Figure \ref{figsap4b} presents the true and estimated IF for example 3. Figures \ref{figsap7} and \ref{figsap8} show the estimated IF for the remaining 9 replications of examples 1 and 2, respectively, from Section \ref{subsecsens}.

\begin{figure}[h!]
	\centering
		\includegraphics[width=0.6\linewidth]{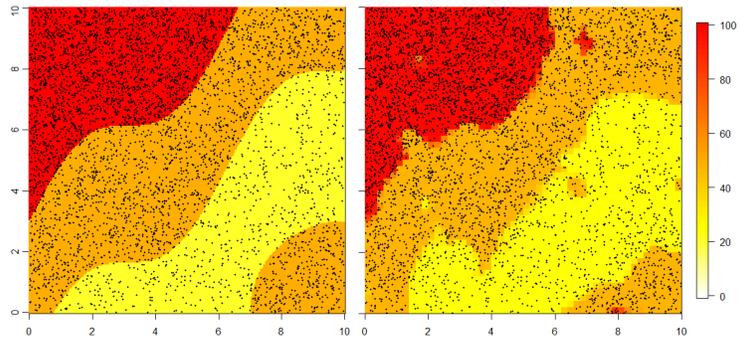}
		\caption{True and posterior mean of the IF of example 3.}
		\label{figsap4b}
\end{figure}

\begin{figure}[h!]
	\centering
		\includegraphics[width=0.8\linewidth]{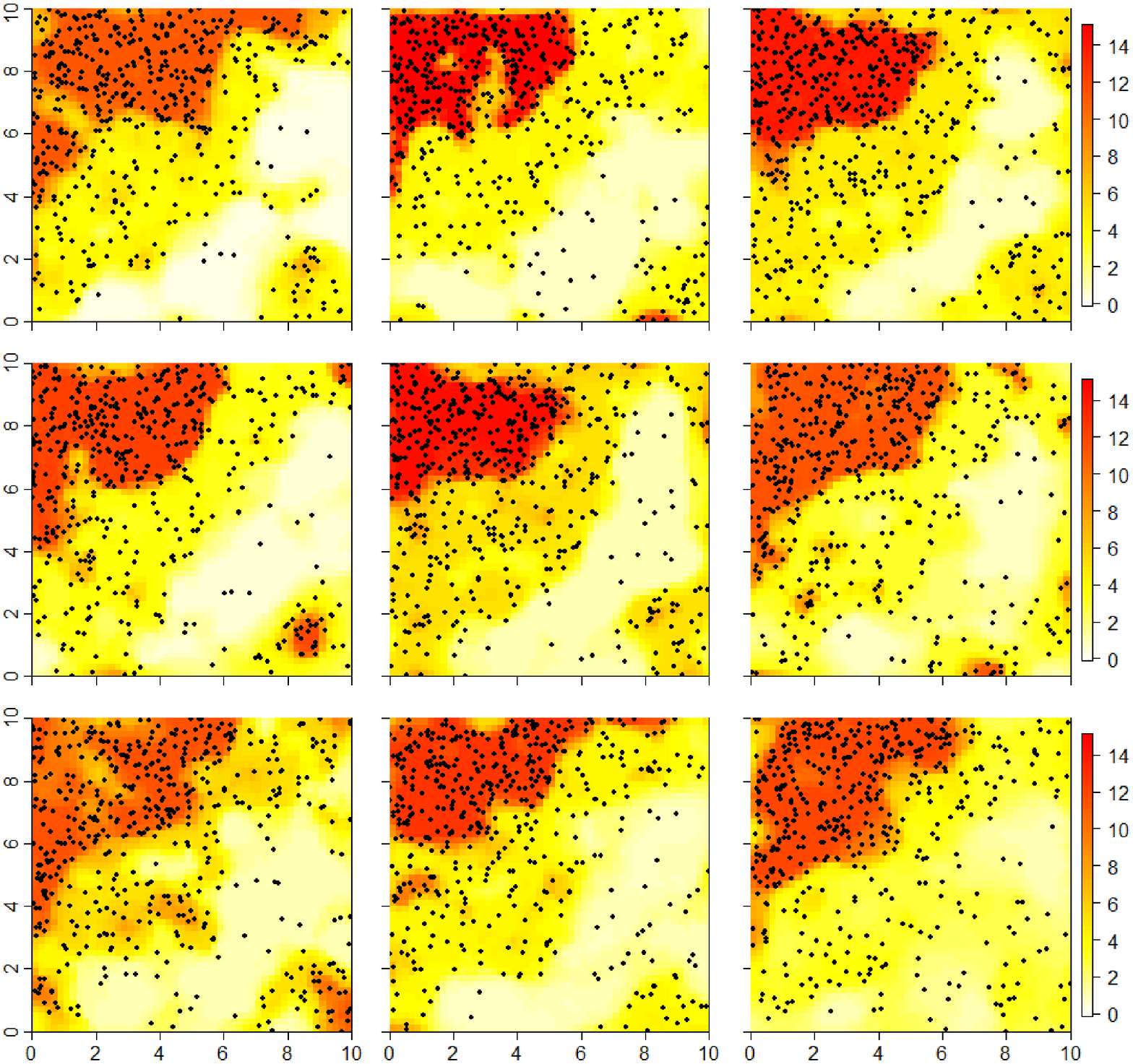}
		\caption{Posterior mean of the IF of the replications of example 1.}
		\label{figsap7}
\end{figure}

\begin{figure}[h!]
	\centering
		\includegraphics[width=0.8\linewidth]{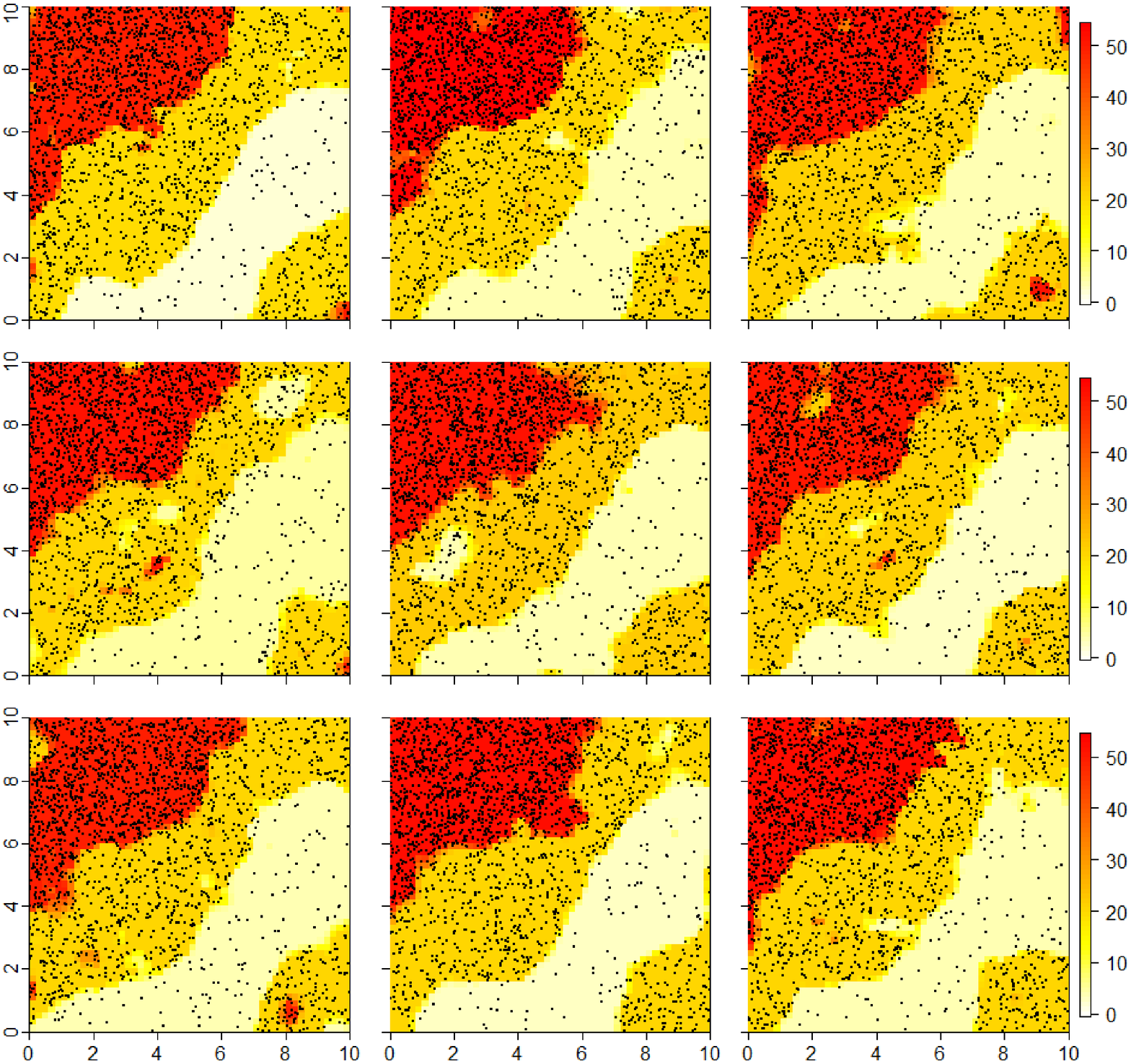}
		\caption{Posterior mean of the IF of the replications of example 2.}
		\label{figsap8}
\end{figure}

Figure \ref{figsap10} shows the estimated IF obtained by applying the methodology from \cite{geng21} for three levels of discretization.

\begin{figure}[h!]
	\centering
		\includegraphics[width=0.8\linewidth]{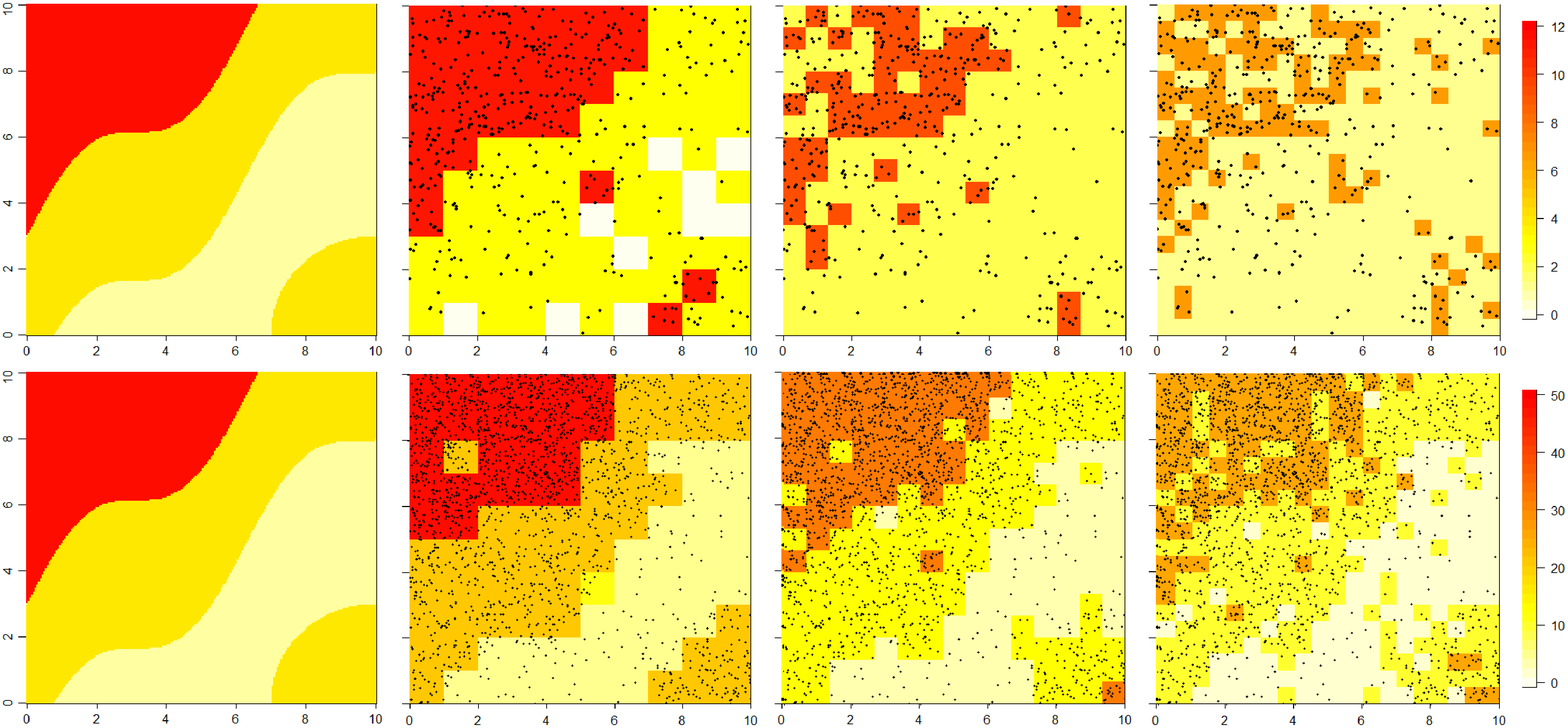}
		\caption{Top: Posterior mean of the IF for example 2 obtained with \cite{geng21}' methodology for discretizations 10x10, 15x15, 20x20. Bottom: Posterior mean of the IF for the New Brunswick fines example obtained \cite{geng21}'s methodology for discretizations 10x10, 15x15, 20x20.}
		\label{figsap10}
\end{figure}

\end{document}